\documentclass[%
 reprint,
 amsmath,amssymb,
 aps,pra]{revtex4-1}
\usepackage{epsfig,amssymb,amsmath,mathrsfs,amsthm,graphicx,float,color}
\usepackage[colorlinks=true]{hyperref}
\usepackage[active]{srcltx}
\usepackage[matrix,frame,arrow]{xypic}
\usepackage[matrix,frame,arrow]{xy}
\usepackage{amsmath}

\newcommand{\qw}[1][-1]{\ar @{-} [0,#1]}

\newcommand{\gate}[1]{*{\xy *+<.6em>{#1};p\save+LU;+RU **\dir{-}\restore\save+RU;+RD **\dir{-}\restore\save+RD;+LD **\dir{-}\restore\POS+LD;+LU **\dir{-}\endxy} \qw}

\newcommand{\multimeasureD}[2]{*+<1em,.9em>{\hphantom{#2}}\save[0,0].[#1,0];p\save !C *{#2},p+LU+<0em,0em>;+RU+<-.8em,0em> **\dir{-}\restore\save +LD;+LU **\dir{-}\restore\save +LD;+RD-<.8em,0em> **\dir{-} \restore\save +RD+<0em,.8em>;+RU-<0em,.8em> **\dir{-} \restore \POS !UR*!UR{\cir<.9em>{r_d}};!DR*!DR{\cir<.9em>{d_l}}\restore \qw}

\newcommand{\multigate}[2]{*+<1em,.9em>{\hphantom{#2}} \qw \POS[0,0].[#1,0];p !C *{#2},p \save+LU;+RU **\dir{-}\restore\save+RU;+RD **\dir{-}\restore\save+RD;+LD **\dir{-}\restore\save+LD;+LU **\dir{-}\restore}
\newcommand{\ghost}[1]{*+<1em,.9em>{\hphantom{#1}} \qw}

\newcommand{\Qcircuit}[1][0em]{\xymatrix @*[o] @*=<#1>}  
 \renewcommand{\Qcircuit}[1][0em]{\xymatrix @*=<#1>}

\newcommand{\pureghost}[1]{*+<1em,.9em>{\hphantom{#1}}}
\newcommand{\multiprepareC}[2]{*+<1em,.9em>{\hphantom{#2}}\save[0,0].[#1,0];p\save !C
  *{#2},p+RU+<0em,0em>;+LU+<+.8em,0em> **\dir{-}\restore\save +RD;+RU **\dir{-}\restore\save
  +RD;+LD+<.8em,0em> **\dir{-} \restore\save +LD+<0em,.8em>;+LU-<0em,.8em> **\dir{-} \restore \POS
  !UL*!UL{\cir<.9em>{u_r}};!DL*!DL{\cir<.9em>{l_u}}\restore}

\newcommand{\poloFantasmaCn}[1]{{{}^{#1}_{\phantom{#1}}}}

\usepackage{graphicx,subfigure}
\usepackage[normalem]{ulem}

\expandafter\let\csname equation*\endcsname\relax

\expandafter\let\csname endequation*\endcsname\relax

\usepackage{subfigure}
\newcommand{\set}[1]{\mathsf{#1}}
\newcommand{\grp}[1]{\mathsf{#1}}
\newcommand{\spc}[1]{\mathcal{#1}}

\def\d{{\rm d}}

\newcommand{\Span}{{\mathsf{Span}}}

\def\>{\rangle}
\def\<{\langle}

\newcommand{\st}[1]{\mathbf{#1}}

\newcommand{\map}[1]{\mathcal{#1}}

\theoremstyle{remark}

\newtheorem{lemma}{Lemma}
\newtheorem{prop}{Proposition}

\def\Tr{\operatorname{Tr}}

\begin{document}

\preprint{APS/123-QED}

\title{Quantum speedup in the identification of cause-effect relations
}
\author{Giulio Chiribella$^{*1,2,3}$ and Daniel Ebler$^{4,1}$} 
\affiliation{
$^1$ Department of Computer Science, The University of Hong Kong, Pokfulam Road, Hong Kong}
\email{giulio@cs.hku.hk}
\affiliation{
$^2$  Department of Computer Science, University of Oxford, Oxford, OX1 3QD, United Kingdom\\
$^3$  Perimeter Institute for Theoretical Physics, Waterloo, Ontario N2L 2Y5, Canada}
\affiliation{
$^4$  Institute for Quantum Science and Engineering,  Department of Physics, Southern University of Science and Technology, Shenzhen, China.}

\begin{abstract}
The ability to identify cause-effect relations is an essential  component of the scientific method. 
The identification of causal relations  is generally accomplished through  
statistical trials where alternative hypotheses are tested against each other. 
Traditionally, such trials have been based on classical statistics. However, classical statistics  becomes inadequate at the quantum scale, where a richer spectrum of causal relations is accessible.
 Here we show that quantum strategies can greatly speed up the identification of causal relations. 
 We analyse the  task of identifying the effect of a given variable, and we show that the optimal quantum strategy beats all classical strategies by running multiple equivalent tests in a quantum superposition. The same working principle leads to advantages in the  detection of  a causal link between two  variables, and in the identification of the cause of a given variable. 
   \end{abstract}
\maketitle

\section*{Introduction}
  Identifying  causal relations  is a fundamental primitive in a variety of  areas,  including  machine learning, 
  medicine,   and genetics   \cite{spirtes2000causation,pearl2009,pearl2014}. 
    A canonical approach is to formulate different hypotheses on the cause-effect relations characterising a given phenomenon, and test them against each other. 
    For example, in a drug test  some patients are administered the drug, while others are administered a placebo, with the scope of determining whether or not the drug causes recovery.  Traditionally, causal discovery techniques have been based on classical statistics, which effectively describes the behaviour   of  macroscopic variables.   
     However, classical techniques become inadequate when dealing with quantum systems,  whose response to interventions can strikingly differ from that of classical random variables  \cite{chaves2018quantum,van2018quantum}.

 Recently, there has been a growing interest in the extension of causal reasoning to  the quantum domain.   Several  quantum generalizations of the notion of causal network  have been proposed     \cite{leifer2006quantum,chiribella-dariano-2009-pra,coecke2012picturing,leifer2013towards, henson2014theory, pienaar2015graph,costa2016quantum,portmann2017causal,allen2017quantum,maclean2017quantum}  and new  algorithms for quantum causal discovery have been designed \cite{wood2015lesson,fitzsimons2015quantum, ried2015, chaves2015information, giarmatzi2018quantum}.   Besides its foundational relevance, the study of quantum causal discovery algorithms is expected to have applications in the emerging area of quantum machine learning \cite{schuld2015introduction,biamonte2017quantum}, in the same way as  classical causal discovery algorithms have previously impacted  classical artificial intelligence.
   
   An intriguing possibility  is that quantum mechanics may provide enhanced ways   to identify   causal links.    A  clue in this direction comes from Refs. \cite{fitzsimons2015quantum,ried2015}, where the authors  show that  certain quantum  correlations are witnesses of  causal relationships, in  apparent violation of the classical tenet   \textit{correlation does not imply causation}.    This observation suggests that quantum setups for testing causal relationships could overcome some of the limitations of existing classical setups. However, the type of advantage highlighted in \cite{fitzsimons2015quantum,ried2015} only concerns a limited class of setups, where the  experimenter  is constrained  to a subset of the possible interventions.    If  arbitrary interventions are allowed, this particular type of advantage disappears.   A fundamental open question is whether   quantum setups can offer an advantage over all classical setups, without  any restriction  on the experimenter's interventions.

Here  we answer the question in the affirmative, proving that quantum features like superposition and entanglement can significantly speed up  the identification  of  causal relations.   We start from the task of  deciding which variable, out of a list of candidates,  is the effect of a given variable.    We first analyze the problem  in the classical setting,  determining the performance of the best classical strategy. Then, we construct a quantum strategy that reduces the error probability by an exponential amount, doubling the decay rate of the error probability with the number of accesses to the relevant variables. 
Remarkably, the decay rate of  our strategy is  the highest achievable rate allowed by quantum mechanics, even if one allows for  exotic setups where the order of operations is indefinite  \cite{chiribella2013quantum,oreshkov2012quantum}.      

The key ingredient of the quantum  speedup   is the ability to run multiple equivalent experiments in a quantum superposition.  
 The same working principle enables  quantum speedups in a broader set of tasks, including, {\em e.g.}, the task of deciding whether there exists a causal link between two given variables, and the task of identifying the cause of a given variable. 
    
    \section*{Results}
    \setcounter{equation}{1}
 {\em Theory-independent framework for testing causal hypotheses.}  Here we outline a  framework for testing causal hypotheses in general  physical theories  \cite{hardy2001quantum,barnum2007generalized,barrett2007information,chiribella2010probabilistic,hardy2011foliable,chiribella2016quantumtheory}. In this framework, variables are represented as physical systems, each system with its set of states.   The framework applies to   theories satisfying the Causality Axiom \cite{chiribella2010probabilistic}, stating that  the probability of an event at a given time should not depend on  choices of settings made at future times. 
  
   A causal  relation between  variable $A$ and variable $B$ is represented by a  map   describing how the state of  $B$ responds to changes in the state of  $A$.
       If the map discards $A$ and outputs a fixed  state of $B$,   then no causal influence  can be observed.    
     In all the other cases, some change of $A$ will lead to an observable change of $B$. Hence, we say that $A$ is a cause for $B$.  
     
In general, the set of allowed causal relationships depends on the physical theory, which determines which maps can be implemented by physical processes. 
   In classical physics,  cause-effect relations   can be represented by  conditional probability distributions of the form $p(b|a)$, where $a$ and $b$ are the values of the random variables $A$ and $B$, respectively.  In quantum theory, cause-effect relations   are described by quantum channels, {\em i.e.} 
 completely positive trace-preserving maps transforming density matrices of system $A$ into density matrices of system $B$.    

Given a set of variables, one can formulate   hypotheses on the causal relationships  among them. For example,   consider a three-variable scenario, where  variable $A$ may cause either variable $B$ or  variable $C$, but not both.  The causal relation is described by a process $\map C$, with input $A$ and outputs $B$ and $C$. Here we consider two alternative causal hypotheses: either  \textit{$A$ causes $B$ but not $C$},   or   \textit{$A$ causes $C$ but not $B$}. 
   The problem  is to distinguish between these  two hypotheses without having further  knowledge of the  physical process responsible for  the causal relation.  
  This means that the process  $\map C$   is unknown, except for the fact that it must compatible with one and only one of the two hypotheses.    Mathematically, the two hypotheses correspond to two sets of physical processes, and   the problem is to determine which set contains the process $\map C$.

In order to decide which hypothesis is correct, we assume that the experimenter has  black box access to the physical process $\map C$. 
The experimenter   can probe the process for $N$ times, intervening between one instance and the next, as illustrated in Figure   \ref{fig:sequential}. In the end, a measurement is performed and its outcome is used to guess the correct  hypothesis.   

\medskip

               \begin{figure}[ht]
        	\centering
        	\includegraphics[width=0.5\textwidth]{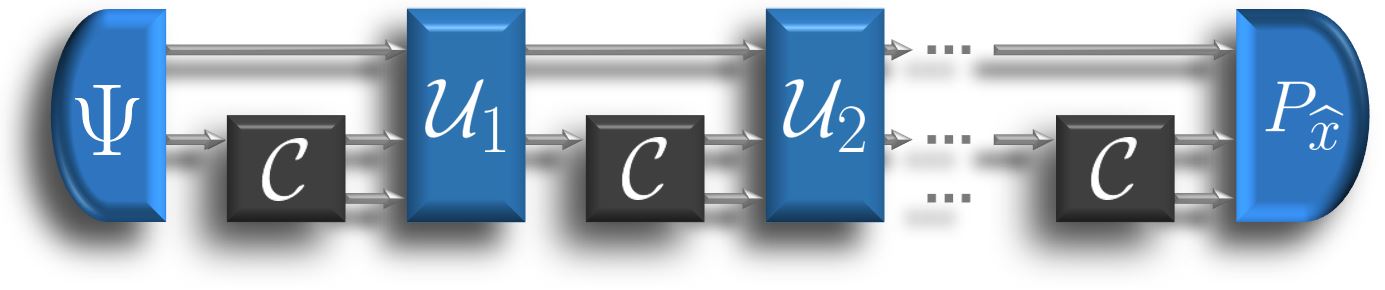}
        	\caption{\footnotesize
        		\textbf{Testing causal hypotheses in the black box scenario.}   The unknown process $\map C$ induces a causal relation between one input variable and two output variables.    The experimenter probes the process for $N$ times, intervening on the relevant  variables at each time step.   The first intervention is the preparation of a state  $\Psi$, involving the  input of the black box and, possibly, an additional reference system (top wire). The subsequent interventions  $\map U_i$   manipulate the output variables and prepare the inputs variables for the next steps.  In the end, the output variables and the reference system are measured, and the measurement outcome is used to  infer  the  causal relation.  }
        	\label{fig:sequential}
        \end{figure}

 An important question is how fast the probability of error decays with  $N$.  
  The decay is typically exponential, with an error probability  vanishing as $p_{\rm err}  (N)    \approx  2^{-R N}$ for some positive constant $R$, which we
     call the {\em discrimination rate}.    
 The operational meaning of the discrimination rate is the following. Given an error threshold $\epsilon$,  the error probability can be made smaller than $\epsilon$ using approximately  $N  > \log\epsilon^{-1}/R$ calls to the unknown process.   The bigger the rate, the smaller the number of calls needed to bring the error below the desired threshold. 

Since the explicit form of the process $\map C$ is unknown, we take $p_{\rm err} (N)$ to be the worst-case  probability over all processes compatible with the two  given causal hypotheses.  If prior information over $\map C$ is available, one may also consider a weaker performance measure, based on  the average  with respect to some prior.  In the following  we stick to the worst case scenario,  as it provides a stronger guarantee on the performance of the test.

\medskip

               \begin{figure}[ht]
        	\centering
        	\includegraphics[width=0.2\textwidth]{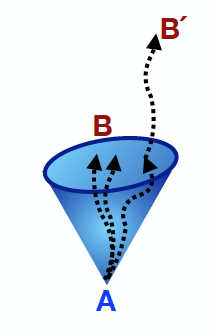}
        	\caption{\footnotesize
        		\textbf{Spacetime picture of a causal intermediary.}   Variable $A$ is localized at a point in spacetime, and its causal influences propagate within its future light cone.  Variable $B$ is distributed over a section of the light cone of $A$ and intercepts all the influences of $A$. Every other variable  $B'$ that is affected by $A$ and comes after $B$ must  be obtained from variable $B$ through some physical process.  
		}    
        	\label{fig:intermediary}
        \end{figure}

{\em Identifying causal intermediaries.}
  A variable  $B$ is  a  causal intermediary for variable $A$ if  all the influences of $A$ propagate through $B$.     Physically, one can think of  $B$ as a slice of the future light cone of $A$, so that all  causal influences of $A$ must pass through $B$, as illustrated in Figure \ref{fig:intermediary}.     Mathematically,  the fact that $B$ is a causal intermediary means that there exist a process $\map C$ from $A$ to $B$ such that  
for every other variable $B'$ and for every  process  $\map C'$ with input  $A$ and output $B'$   one can decompose $\map C'$ as $\map C'  =  \map R\circ \map C$, where $\map R$ is a suitable process from $B$ to $B'$. 

The condition that a variable is a causal intermediary of another has a simple characterisation in all  physical theories where processes are fundamentally reversible, meaning that they can be modelled as the result of a reversible evolution  of the system and an environment \cite{chiribella2010probabilistic}.   The reversibility condition is captured by the expression  $ \map C  =  (\map I_{B}  \otimes  \Tr_{E'})   \map U  (\map I_A \otimes \eta_E) $, 
 where variables $E$ and $E'$ represent the environment (before and after the interaction), $\eta$ is the initial state of the environment,  $\Tr_{E'}$ is the operation of discarding system $E'$ \cite{chiribella2010probabilistic},  and  $\map U$ is a reversible process from $AE$ to $BE'$.  

When the reversibility condition  is satisfied, the variable $A$ can be  recovered  from variables $B$ and $E'$.  If variable $B$ is to be a causal intermediary of $A$, then the process $\map C$ must be correctable, in the sense that its action can be undone by another process $\map R$.  
In addition, if the state spaces of variables $A$ and $B$ are finite dimensional and of the same dimension, then the process $\map C$ must be  reversible.   In classical theory, this means that $\map C$ is an invertible function. In quantum theory, this means that $\map C$ is a unitary channel, of the form $\map C (\rho)  =  U \rho U^\dag$ for some unitary operator $U$.   

In the following, we will consider the  task of identifying which variable, out of a given set of candidates,  is the causal intermediary of  a given variable $A$.  An important feature of this task is that it admits a complete analytical treatment, allowing us to rigorously prove a quantum   advantage over  all classical strategies.   Besides its fundamental interest, this advantage could have applications to the task of monitoring the information flow in  future quantum communication networks, allowing an experimenter to determine which node of a quantum network receives information from a given source node.  

\medskip 

{\em Optimal classical strategy.}    Suppose that   $A$, $B$, and $C$ are  random variables with the same alphabet  of size   $d<\infty$.    In this case, the fact that $X  \in  \{  B,C\}$ is a causal intermediary for $A$ means that the map from $A$ to $X$ is a permutation.  The first (second) causal  hypothesis is that $B$ ($C$) is  a permutation of $A$, while  $C$ ($B$) is  uniformly random.    Other than this, no information about the functional relation between the variables is known to the experimenter. In particular, the experimenter does not know which permutation  relates the variable $A$ to its causal intermediary $X$.  

Let us determine how well one can distinguish between the two hypotheses with  a finite number of experiments.     In principle, we should examine all sequential strategies  as in Figure \ref{fig:sequential}.  However, in classical theory the problem can be greatly  simplified: the optimal discrimination rate can be achieved by a parallel strategy, wherein the $N$ input variables are initially set  to some prescribed set of values \cite{hayashi2009discrimination}.  

    The possibility of an error arises  is when the randomly fluctuating variable   accidentally takes values that are compatible with a permutation, so that the outcome of the test gives no ground  to discriminate between the two hypotheses.   The probability of such inconclusive scenario  is equal to  $P(d,v)/d^{N}$, where $v$ is the number of distinct values of $A$ probed in the experiment and $P(d,v)=  d!/(d-v)!$ is the number of injective functions from a $v$-element set to a $d$-element set.   The probability of confusion is minimal for $v=1$,  leading to the overall error probability  
\begin{align}\label{cprob}
p_{\rm err}^{\rm C}  =  \frac1 {2d^{N-1}} \, .
\end{align}
  As a consequence, the rate at which the two causal hypotheses can be distinguished from each other is  
  \begin{align}\label{classical} R_{\rm C} =
   \log d \, . 
  \end{align}     

{\em A first quantum advantage.}   Classical systems can be regarded  quantum systems that     lost coherence  across the states of a fixed  basis, consisting of the classical states.    But what if coherence is preserved?    Could a  coherent superposition of classical states be a better probe for the causal structure?     

If the  causal relations are restricted to reversible gates that permute the  classical states,  coherence  offers an immediate advantage.  The experimenter  can prepare $N$ probes, each  in  the superposition   $|e_0\>  =  \sum_{i=0}^{d-1}  |i\>  /\sqrt d$.    Since the superposition is invariant under permutations,  the unknown process will produce  either  $N$ copies of the state $ |e_0\>\<e_0| \otimes I/d$ or $N$ copies of the state $I/d\otimes |e_0\>\<e_0|$, depending on which causal hypothesis holds. Using  Helstrom's minimum error measurement \cite{helstrom1969quantum}, the error probability  is reduced to  
\begin{align}\label{pcoh}
p_{\rm err}^{\rm coh }  = \frac1  {2d^N}  \, .
\end{align}  Compared with the  classical error probability  (\ref{cprob}), the error probability of this simple quantum strategy is reduced  by a factor $d$, which does not change the rate, but could be significant when the size of the alphabet  is large.  

Let us consider the full quantum version of the  problem. 
Three  quantum variables $A,B,$ and $C$, corresponding to $d$-dimensional quantum systems, are promised to satisfy one of two causal hypotheses: either  {\em (i)} the state of $B$  is obtained from the state of $A$ through an arbitrary unitary evolution and the state of $C$  is maximally mixed, or {\em (ii)} the state of $C$  is obtained from the state of $A$ through an arbitrary unitary evolution and the state of $B$  is maximally mixed.

Despite the fact that now  the cause-effect relation can be one of infinitely many unitary gates,  it turns out that  the error probability (\ref{pcoh}) can still be attained.  A universal quantum strategy, working for arbitrary unitary gates, is  to  prepare $d$ particles in the singlet state 
\begin{align}   |S_d\>    =  \frac 1 {\sqrt d!}    \sum_{k_1,k_2, \cdots,  k_d}  \,  \epsilon_{k_1k_2\dots k_d}   \,   |k_1\>  |k_2\>  \cdots |k_d\> \, \label{singlet}
\end{align} 
where $\epsilon_{k_1k_2\dots k_d}$ is the totally antisymmetric tensor and the sum ranges over all vectors in the computational basis. Then,  each of the $d$ particles is used as an input to one use of the channel.
Repeating the experiment for $t$ times, and performing Helstrom's minimum error measurement one can  attain the error probability  $p_{\rm err}^{\rm coh }  =(2d^N)^{-1}$, with $N  =  t d$,  independently of the unitary gate representing the cause-effect relationship.    In summary,  the quantum error probability is at least $d$ times smaller than the best classical error probability, even if the cause-effect relationship is described by an arbitrary unitary gate. 

  \medskip 
   
\begin{figure}
                \centering
                \includegraphics[width=0.3\textwidth]{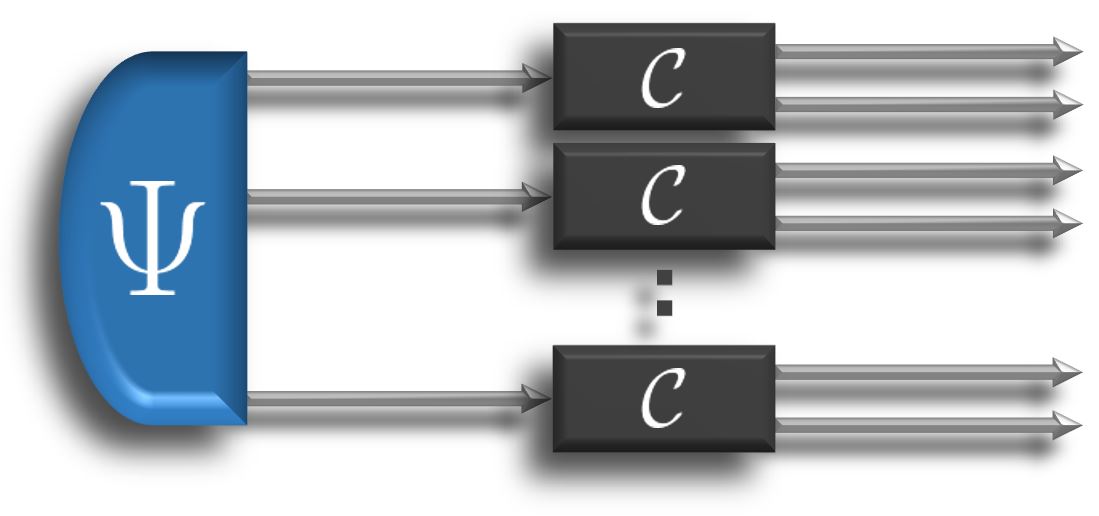}
               \caption{{\bf  Simple parallel strategies.}  The unknown process $\map C$ is probed for  $N$ times, acting in parallel on $N$ identical systems, initially prepared in a correlated state  $\Psi$.  }\label{fig:parallel}
\end{figure} 

{\em Optimality among simple parallel strategies.}   We now show that the value (\ref{pcoh}) is optimal among all  simple strategies where  the unknown process is applied $N$ times in parallel on $N$ identical input systems, as in Figure \ref{fig:parallel}.
 
 Optimality follows from a complementarity relation between the information  about the causal structure and the information about the  functional  dependence between  cause and  effect.  Suppose that  the cause-effect dependence amounts to a unitary gate $U$ in some finite set $\set U$.    The ability of a state $|\Psi\>$ to probe the cause-effect dependence  can be quantified by the probability $p^{\set U}_{\rm guess}$ of correctly guessing the  unitary $U$ from the state $U^{\otimes N } |\Psi\>$. 
  When the set of possibly unitaries has sufficient symmetry, we find that the probability of error in identifying the causal structure satisfies the lower bound 
\begin{align}\label{comple}
p_{\rm err}    \ge  \frac 1{2d^N}    \left\{ 1    +     \frac 1{2  (d^N-1)}  \left(  \frac{  p^{\set U}_{\rm guess}  - \frac 1{|\set U|}  }{ \frac 1  {|\set U|}}   \right)^2     \right\} \ 
\end{align}(Appendix \ref{app1}). 
The higher  the probability of success in  guessing the cause-effect dependence, the higher  the probability of error in identifying the causal structure. 
A consequence of the bound (\ref{comple}) is that the minimum error probability in identifying the causal intermediary  is  $(2d^N)^{-1}$, and is attained  when the success probability  $p^{\set U}_{\rm guess}$ is equal to the random guess probability  $1/|\set U|$.

\medskip

{\em  Exponential reduction of the error probability.} The bound (\ref{comple}) shows that the discrimination rate of simple parallel strategies cannot exceed the classical discrimination rate $\log d$. 
  We now show that that the rate can be doubled  by entangling the $N$ probes with an additional reference system.

  \begin{figure}
        \centering
\includegraphics[width=0.47\textwidth]{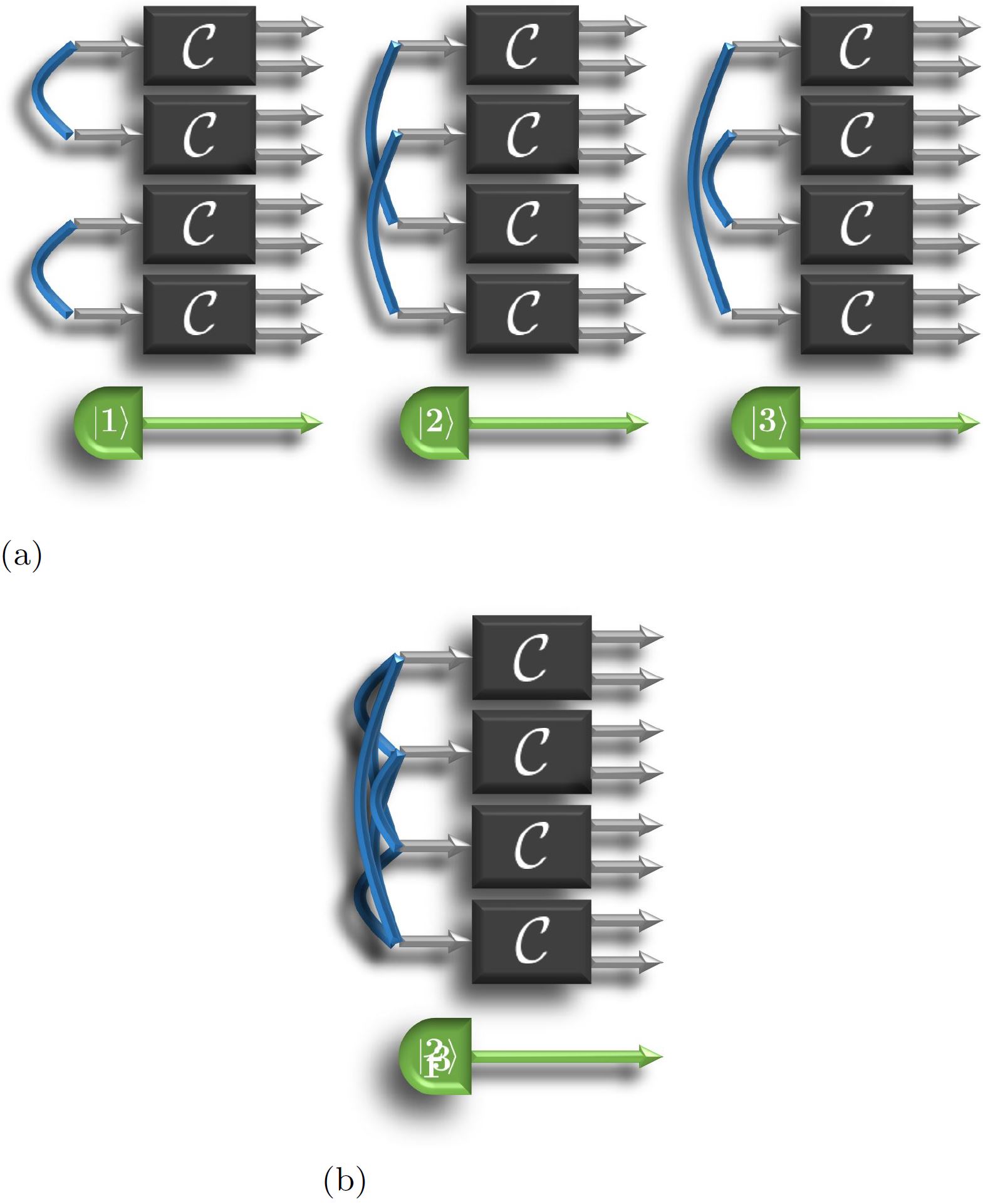}

        \caption{{\bf   Coherent superposition of configurations.}  Subfigure (a) shows the three different ways of dividing four quantum bits into groups of two. These three configurations are all equivalent for the identification of the causal intermediary.   Subfigure (b)  pictorially illustrates  a quantum superposition of  configurations, with the choice of  configuration  correlated with the state of a control system.  }\label{fig:superpos}
\end{figure}

The working principle of our strategy is to build a quantum superposition of equivalent experimental setups.    If no reference system is used,  we know  that the optimal  strategy is to divide the $N$ probes into  $N/d$  groups (assuming for simplicity that $N$ is a multiple of $d$), and to entangle the probes within each group.   Clearly,  different ways of dividing the $N$ inputs into groups of $d$   are equally optimal: it does not matter which particle is  entangled with which, as long as all each particle is part of a singlet state.   Still, we can imagine a machine that partitions the particles according to a certain  configuration $i$ if a control system is in the state $|i\>$. When the control system is in a superposition, the machine will probe the unknown process in a superposition of configurations, as pictorially illustrated in Figure  (\ref{fig:superpos}).   Explicitly, the optimal input state is 
 \begin{align}\label{optimal}
|\Psi \> =    \frac 1 {\sqrt {G_{N,d}}}  \,  \sum_{i=1}^{G_{N,d}}  \,     \left( |S_d\>^{\otimes N/d}\right)_i \otimes |i\>  \, ,
\end{align}   
where $i$ labels  the different ways to partition $N$ identical objects into groups of $d $ elements,   $G_{N,d}$ is the number of such ways,  $ \left( |S_d\>^{\otimes N/d}\right)_i$ is the product of $N/d$  singlet states arranged according to the $i$-th configuration,  and $\{  |i\>  \, , \,  i  = 1,\dots,  G_{N,d}\}$ are orthogonal states of the reference system. 


Classically, there would be  no point in randomizing  optimal configurations, because mixtures cannot reduce the error probability.  But in the quantum case, the coherent superposition  of equivalent configurations brings the error probability down  to 
\begin{align}
 p^{\rm Q}_{\rm{err}}  (r)  &=\frac{r}{2d^N }\left(1-\sqrt{1-   r^{-2}} \right)     \xrightarrow{r\gg 1}  \frac1{4 r  d^N  } \, , \label{quantum1}
\end{align}
 where $r$ is the number of linearly independent states of the form $( |S_d\>^{\otimes  N/d})_i$ (Appendix \ref{app2}).   

To determine how much the error probability can be reduced, we only need to evaluate the  number of linearly independent states. 
     It turns out that this number  grows as $d^N$, up to a polynomial factor (Appendix \ref{app2} again).  Taking the logarithm, we obtain the discrimination   rate  
 \begin{align} R_{\rm Q}  =  -\lim_{N\to \infty}   \frac{ \log p^{\rm Q}_{\rm err}}N  =     2 \log d \, , \label{qmrate}
 \end{align}
 which is  twice the classical discrimination rate (\ref{classical}).   In fact, the asymptotic regime is already reached with a small number of interrogations, of the order of a few tens.   For example,  the causal relation between two quantum bits can be determined with an error  probability smaller than $10^{-6}$  using with 12 interrogations, whereas  20 interrogations are necessary for classical binary variables.

The above strategy  is universal, in that it applies to causal relationships described by arbitrary unitary gates. In particular, it applies to   
gates that permute the classical states.  Hence,  the ability to maintain coherence across the classical states and to generate entanglement with a reference system offers an exponential speedup with respect to the best classical strategy.   In passing, we note that the universal quantum strategy  is insensitive to the presence of perfectly correlated noise, such as the noise due to the lack of a reference frame \cite{bartlett2007reference},  where each of the $N$ input variables is subjected to the same unknown  unitary gate.

\medskip 

{\em  The ultimate quantum limit.}  So far, we examined  strategies where the unknown process is applied in parallel to a large entangled state. 
  Could a general sequence of interventions  achieve  an even better rate? 
  
  Finding the optimal sequential strategy is generally a hard problem. 
   To address this problem, we introduce the  {\em fidelity divergence} of two quantum  channels   $\map C_1$ and $\map C_2$, defined as 
  \begin{align}\label{fid}
\partial F  (\map C_1, \map C_2)  = \inf_R \inf_{\rho_1, \rho_2}  \frac  {  F \Big [ (  \map C_1\otimes \map I_R)  (\rho_1) \, ,    (\map C_2\otimes \map I_R)  (\rho_2)  \Big]}{F( \rho_1,\rho_2)} \, ,
\end{align}  
where  $\rho_1$  and $\rho_2$ are  joint states of the channel's input and of  the reference system $R$. It is understood that the infimum  in the right hand side is taken over pairs of states $(\rho_1, \rho_2)$ for which the fidelity $F(\rho_1,\rho_2)$ is  non-zero,  so that the expression on the right hand side of Equation (\ref{fid}) is well-defined.   

The fidelity divergence quantifies the ability of  channels  $\map C_1$ and $\map C_2$  to  move two states apart from each other.   In the Methods section, we show that the error probability in distinguishing between  $\map C_1$ and $\map C_2$  with  $N$ queries is lower  bounded as 
\begin{align}
p^{\rm seq}_{\rm err}  (\map C_1,\map C_2; N)  \ge  \frac {\partial F   (\map C_1, \map C_2)^N  }4 \, .
\end{align}
  
  In particular, suppose that the two  channels  $\map C_1$ and $\map C_2$ have the form  $\map C_1  =  \map U\otimes I/d$ and $\map C_2  = I/d \otimes \map U$,  where $\map U$ is a fixed unitary channel. In this case,  we find that   the fidelity divergence is  $1/d^2$.  
 Hence, the error probability satisfies the bound
 \begin{align}\label{ratebound}
p^{\rm seq}_{\rm err}  (\map C_1,\map C_2; N)  \ge  \frac {1  }{4 d^{2N}}\, .
\end{align} 

In the causal intermediary problem, the unitary gate $\map U$ is unknown, and therefore the error probability can only be larger than $p^{\rm seq}_{\rm err}  (\map C_1,\map C_2; N) $.  Hence, the identification of the causal intermediary cannot occur at a rate faster than $2 \log d$.   

Equation (\ref{ratebound}) limits all sequential quantum strategies. But in fact   quantum theory is also  compatible with scenarios where physical processes take place in an indefinite order  \cite{chiribella2013quantum,oreshkov2012quantum}. Could the rate be increased if the experimenter had access to exotic phenomena involving indefinite order?    
 
The answer is negative.  In the Methods section we develop the concepts and methods needed to answer this question, and we show that  the minimum error probability in distinguishing between the two channels $\map C_1  =  \map I\otimes I/d$ and $\map C_2  = I/d \otimes \map I$  using arbitrary setups with indefinite order  satisfies the bound 
\begin{align}
p_{\rm err}^{\rm ind}   (\map C_1,\map C_2; N)  \ge  \frac{ 1- \sqrt{ 1-  \frac 1 {d^{2N}}}}2  \, . 
\end{align}
Clearly, this bound applies to the causal intermediary problem, which is  harder than the discrimination of the two specific channels $\map C_1  =  \map I\otimes I/d$ and $\map C_2  = I/d \otimes \map I$.    Hence, the rate $R_{\rm Q}  =   2 \log d$ represents the ultimate quantum limit to the identification of a causal intermediary.       
 
\medskip 

{\em Extension to arbitrary numbers of hypotheses.}  
The quantum advantage demonstrated in the previous sections can be extended to  the identification of the causal intermediary among an arbitrary number  $k$ of candidate variables. 
 The best classical strategy  still consists   in initializing all variables to the same value.  Errors arise when  the  values of two or more output variables are compatible with an invertible function.   In the limit of many repetitions, the minimum  error probability is    $ p^{\rm C}_{{\rm err}, k}   =   (k-1)/(2d^{N-1}) +  O\left({d^{-2 N}}\right)$.  
  (Appendix \ref{app3}). 
For  quantum strategies, the best option among simple parallel strategies  is still to divide the input particles into $N/d$ groups of $d$ particles and to initialize each group in the singlet state. In  Appendix \ref{app4}, we show that this strategy reduces the error probability to
$ p^{\rm coh}_{{\rm err},  k }  =   (k-1)/(2d^N) +  O\left( {d^{-2N} }\right)$, for causal relations represented by arbitrary unitary gates.  

An exponentially smaller error probability  can be achieved using  the input state (\ref{optimal}).    The evaluation of the error probability is  more complex than in the two-hypothesis case, but the end result is the same: when the causal dependency is probed $N$ times, the quantum error probability  decays at the exponential rate $R_{\rm Q}  =  2 \log d$, twice the rate of the best classical strategy (see Appendix \ref{app5} for the technical details).

\medskip

{\em Applications to other tests of causal hypotheses. }  The strategies developed  in the previous sections can be applied to the identification of causal relations in a variety of  scenarios.    For example, they can be used to decide whether there is a causal link between two variables $A$ and $B$. More specifically, they can be used to determine whether variable $B$  is a causal intermediary for   variable $A$ or whether $B$ fluctuates at random independently of $A$.  Also in this case, the error probability of the best classical strategy is $1/(2d^{N-1})$, whereas preparing $N/d$ copies of the singlet yields error probability $1/(2d^N)$.  

By superposing all possible partitions of the $N$ inputs into groups of $d$, one can boost the discrimination rate from $\log d$ to $2 \log d$. 
  One could speculate that, in the future,  such a fast identification could be useful as a quantum version of the ping protocol, capable of establishing whether there exists  a quantum  communication link between  two nodes of a quantum internet \cite{kimble2008quantum}. 


Another application  of our techniques is in the problem of identifying the cause of a given variable. Suppose that  one  of $k$ variables $A_1,A_2,\dots, A_k$ is the  cause  for a given variable $B$.    An example of this situation arises in genetics, when  trying to identify the gene responsible  for a certain characteristic. 
  Here, the  interesting scenario is when the number of candidate causes is large.

Classically, the problem is to find the variable $A_x$ such  that $B$ is a function of  $A_x$.    For simplicity, we  first assume  that all variables have the same $d$-dimensional alphabet, and that the function  from $A_x$ to $B$ is the identity, namely $b=  a_x$.   In this case, the cause can be identified without any error by probing the unknown process for $\lceil \log_d   k \rceil$ times. The identification is done by  a simple search algorithm, where one  divides the candidate variables in $d$ groups and initializes the input variables in the $i$-th group to the value $i$. In this way, $d-1$ groups can be ruled out, and one can iterate the search in the remaining group. Using a decision tree argument \cite{cormen2009introduction}, it is not hard to see that $\lceil \log_d   k \rceil$ is the minimum number of queries needed to identify the unknown process in the worst case scenario.  

In the quantum version of the problem, we find that the number of queries can be cut down by  approximately a half when the number of hypotheses is large.  The trick is to prepare $k$ maximally entangled states, and to apply the unknown process to the first system of each pair.  Repeating this procedure for $N$ times and  using results on port-based teleportation \cite{mozrzymas2018optimal}  we find  that the error  probability  is $p_{\rm err}   =  (k-1)/(d^{2N} + k  -  1)$. Hence,  $N    =\lceil  (1+\epsilon) (\log_d k) /2\rceil $ queries are sufficient to identify the cause with vanishing error probability in the large $k$ limit.  

In Appendix \ref{app6}  we consider the more complex scenario where  the functional dependence between the cause and effect is  unknown, and the only assumption is that the effect is a causal intermediary of the cause. Despite the lack of information about the functional dependence, we show that the correct cause can be still identified with high probability using   $N=\lceil (1+\epsilon)  (\log_d   k) /2\rceil$ calls to the unknown process.  
 The fast identification of the cause is achieved by dividing the $N$ copies of each input variable $A_i$ into  groups of $d$ copies, preparing each group in the singlet state, and entangling the configuration of the groupings with an external reference system.  Once again, the superposition of multiple equivalent setups leads to a quantum speedup over the best classical strategy.

\medskip 

\section*{Discussion}   
We showed that quantum mechanics enhances our ability to  detect direct cause-effect links.   This finding motivates the exploration of  more complex networks of  causal relations, including intermediate nodes  and global  causal dependences between groups of variables  \cite{spirtes2000causation,pearl2009,pearl2014}.   The development of new  techniques for testing causal relations  could find applications to future quantum communication networks, providing a fast way to test the presence of communication links. It could also assist the design of intelligent quantum machines, in a similar way as classical causal discovery algorithms  have been useful in classical artificial intelligence. 
 In view of such applications, it is  important to go beyond the noiseless scenario considered in this paper, and to address scenarios  where the cause-effect relationships are obfuscated by  noise.  The techniques developed in our work already provide some insights in this direction. Quite interestingly, one can  show that the quantum advantage persists in the presence of depolarizing noise, provided that the noise level  is not too high (see Appendix \ref{app7}). A complete study of the noisy scenario, however,  remains an open direction of future research.

Another direction of future investigation is foundational.  Given the advantage of quantum theory over classical theory,   it is tempting to  ask whether alternative physical theories could offer even larger advantages.   Interesting candidates  are theories that admit more powerful dense coding protocols than quantum theory \cite{massar2015hyperdense}, as one might expect super-quantum advantages to arise from the presence of stronger correlations with the reference system.  In a similar vein, one could explore physical theories with higher dimensional state spaces, such as Zyczkowski's quartic theory \cite{zyczkowski2008quartic}, or quantum theory on quaterionic Hilbert spaces \cite{barnum2015some}.     Indeed, it is intriguing to observe that the classical rate $R^{\rm C}   = \log d$    and the quantum rate $R^{\rm Q} =  2\log d$ are equal to the logarithms of the dimensions of the classical and quantum state spaces, respectively.  In general, one may expect a relationship between the dimension of the state space and the rate.  Should super-quantum advantages emerge, it would be natural to ask which physical principle determines the  causal identification power of quantum mechanics.  An intriguing possibility is that one of the hidden physical principles of quantum theory could be a principle on the ability to distinguish alternative causal hypotheses.

\medskip
\section*{Methods}  

{\em Properties of the fidelity divergence.} Here we derive two  properties of the fidelity divergence defined in Equation (\ref{fid}).  First, the fidelity divergence provides a lower bound on  the probability of misidentifying  a  channel with another: 
\begin{prop}
The probability of error in distinguishing between two  quantum channels $\map C_1$ and $\map C_2$ with $N$ queries is lower bounded as $p_{\rm err}^{\rm seq} (\map C_1,\map C_2 ;N) \ge \partial F( \map C_1,\map C_2)^N/4$. 
\end{prop}
The bound can be obtained in the following way.  Let $\rho^{(N)}_x$  be the output state of a circuit as in  Figure \ref{fig:sequential}.    Then, we have the bound   
\begin{align}
p^{\rm seq}_{\rm err}  (\map C_1,\map C_2  ;N)  & = \frac{1}{2} \left( 1-\frac{1}{2} \left\| \rho^{(N)}_1 - \rho^{(N)}_2 \right\|_1  \right) \nonumber \\
& \geq \frac{1}{2} \left( 1- \sqrt{1- F\Big(\rho^{(N)}_1 , \rho^{(N)}_2\Big)}\right) \nonumber \\
&\geq \frac{1}{2} \left[ 1- \sqrt{1- \partial F^{N} (\map C_1 , \map C_2)}\right] \nonumber \\
& \geq \frac{1}{2} \left[ 1- \left(1- \frac{ \partial  F^{N} (\map C_1 , \map C_2)}{2} \right)\right] \nonumber \\
&=  \frac {\partial F   (\map C_1, \map C_2)^{N}  }4 \, .  \label{boundfid}
\end{align}
The first line follows  from Helstrom's theorem \cite{helstrom1969quantum}, and the second line follows from the Fuchs-Van De Graaf Inequality \cite{fuchs1999cryptographic}. The third line  follows from the definition of the fidelity divergence  (\ref{fid}),  which implies that the fidelity between the states right after the $(t+1)$-th use of the unknown channel $\map C_x$,  denoted by $\rho_{x,  t+1}$,   satisfies the bound 
\begin{align}
\nonumber F\big(\rho_{1, t+1}  ,  \rho_{2 ,t+1}\big)    &\ge   \partial F  (  \map C_1, \map C_2)  F\big(\map U_{t+1} \rho_{1,t}  ,  \map U_{t+1}\rho_{2, t}\big)  \\
&\ge \partial F  (  \map C_1, \map C_2)  F\big( \rho_{1,t }  , \rho_{2,t} \big) \, ,
\end{align} 
where $\map U_{t+1}$ is the $(t+1)$-th operation in Figure \ref{fig:sequential}. 
The fourth line follows from the elementary inequality $\sqrt{1-t}  \le 1-t/2$.

Another important property is that the fidelity divergence can be evaluated on pure states. 
The proof is simple:   let $\rho_1$ and $\rho_2$ be two arbitrary states of the composite system $AR$, where $R$ is an arbitrary reference system.     By Uhlmann's theorem \cite{uhlmann1976transition}, there exists a third system $E$ and two purifications  $|\Psi_1\> ,|\Psi_2\> \in \spc H_A\otimes \spc H_R\otimes \spc H_E$,   such that   $F (\Psi_1,  \Psi_2)  =  F (\rho_1, \rho_2)$.      On the other hand, the monotonicity  of the fidelity under partial trace \cite{wilde2013quantum},  ensures that the  fidelity between the output states $(  \map C_1\otimes \map I_{RE})  (\Psi_1)$ and $(  \map C_2\otimes \map I_{RE})  (\Psi_2)$ cannot be larger than the fidelity between the states $(  \map C_1\otimes \map I_R)  (\rho_1)$ and $(  \map C_2\otimes \map I_R)  (\rho_2)$.   Hence, the minimization  on the right hand side of equation (\ref{fid}) can be restricted without loss of generality to pure states. 
 
 \medskip 
 
{\em Fidelity divergence for the identification of the causal intermediary.}  
 Let us see how the fidelity divergence  can be applied to our causal identification problem.  The two channels are of the form $\map C_{1,U}  (\rho)  =  U\rho  U^\dag  \otimes I/d  $  and $\map C_{2,V}   =  I/d \otimes V\rho V^\dag$, where $U$ and $V$ are two unknown unitary gates. 
   Since we are interested in the worst case scenario, every choice of $U$ and $V$ will give an upper bound to the discrimination rate.   In particular, we  pick $U=  V $.  
 
 \begin{prop}
 The fidelity divergence for the two channels $\map C_{1,U}$  and $\map C_{2,U} $  is $\partial F  (\map C_{1,U},\map C_{2,U})  =  1/d^2$. 
 \end{prop}  


By the unitary invariance of the fidelity, $\partial F  (\map C_{1,U},\map C_{2,U}) $ is independent of $U$. Without loss of generality, let us pick $U=I$.  
 For a generic reference system $R$ and  two generic pure  states $|\Psi_1\>, |\Psi_2\>  \in  \spc H_A\otimes \spc H_R$, the two output states are 
\begin{align}
\nonumber \rho_1'   & =  (  \map C_{1,I}\otimes \map I_R)  (\Psi_1)   =  (\Psi_1)_{BR} \otimes \frac {I_C}{d}  \\
\rho_2'   & =  (  \map C_{2,I}\otimes \map I_R)  (\Psi_2)   =    \frac {I_B}{d}  \otimes (\Psi_1)_{CR}    \, ,
\end{align} 
up to reordering of the Hilbert spaces.  
 The fidelity can be computed with the relation 
\begin{align}\label{fidsqrt}
F (\rho_1', \rho_2')   &  =  \frac {  \left|  \Tr \left [ \sqrt{  (\Psi_1)_{BR} \, (\Psi_2)_{CR}  \, (\Psi_1)_{BR} } \right]  \right|^2}{d^2} \,,
\end{align}
where we  omitted the identity operators for the sake of brevity.   Let us expand the input states as 
\begin{align}\label{decompo}
|\Psi_x\>    =  \sum_{n}  \,  |\phi_{xn}\> \otimes |n\> \, ,  \qquad x\in \{0,1\}
\end{align} 
where $\{|n\>\}$ is an orthonormal basis for the reference system,  and $\{|\psi_{xn}\>\}$ is a set of unnormalized vectors.  
Inserting Equation (\ref{decompo}) into Equation   (\ref{fidsqrt}), we obtain the expression
\begin{align}
F (\rho_1', \rho_2')   &  =  \frac {  \left|  \Tr \left [ \sqrt{ C^\dag C } \right]  \right|^2}{d^2}   =  \frac {  |\,  \Tr   |C|   \,|^2}{d^2} \, ,
\end{align}
with $C =  \sum_n  \,  |\phi_{1n}\>\<\phi_{2n}|$.  On the other hand, the fidelity between the input states is 
\begin{align}
F(\rho_1,  \rho_2)   =  |\<\Psi_1|\Psi_2\>|^2  =  |\Tr  [  C] |^2 \,.
\end{align}  
  Hence, the fidelity divergence satisfies the bound 
  \begin{align}
 \nonumber  \partial F (\map C_1, \map C_2)      & =  \inf_R \inf_{\rho_1,\rho_2}  \frac {F(\rho_1',\rho_2')}{F(\rho_1,\rho_2)}  \\
 \nonumber   &  = \frac{1}{d^2} \inf_C   \left|   \frac{\Tr  |C|   }{\Tr[C]} \right|^2   \\
    &  \ge \frac 1{d^2} \, ,
  \end{align}
  having used the inequality $|\Tr [C]|  \le \Tr |C| $, valid for every operator $C$.   The inequality holds with the equality sign whenever $C$ is positive. This condition is satisfied, e.g. when the input states $|\Psi_1 \>$ and $| \Psi_2\>$ are identical.   \\

{\em Quantum strategies with indefinite causal order.} In principle,   quantum mechanics  is compatible with  situations where multiple processes are combined in  indefinite   order \cite{chiribella2013quantum,oreshkov2012quantum}.   This suggests that an experimenter could devise new ways to probe  quantum channels, allowing  the relative order among different uses of the same channel to be indefinite.  We call such strategies {\em indefinite testers}. 

Consider the problem of identifying a channel $\map C_x$ from  $N$ uses.  The input resource is the channel  $\map C_x^{\otimes N}$,  representing $N$ identical black boxes that can be arranged in any desired order.  Besides  the product of $N$ independent channels,  the most general class of channels with this property  is the class of  no-signalling channels with $N$ pairs of input/output systems.

  Mathematically, an indefinite tester is a linear map from the set of no-signalling channels to the set of probability distributions over a given set of outcomes.  Equivalently, the tester can be  described by a set of operators $\{T_x\}$, where each operator $T_x$ acts on the Hilbert space $ \bigotimes_i  \, (\spc H^{\rm in}_i  \otimes \spc H^{\rm out}_i)$, where $\spc  H_{i}^{\rm in}$ and  $\spc  H_{i}^{\rm out}$ are the  Hilbert spaces of the input and output system in the $i$-th pair, respectively.       When the test is performed on a no-signalling channel $\map C$, the probability of the outcome $x$ is given by the generalized Born rule $p_x  =  \Tr[ T_x  \,  C ]$, where $C$ is the Choi operator of the channel $\map C$ \cite{choi-1975}.   The normalization of the probabilities 
 \begin{align}\label{testernorm}
 \sum_x \, \Tr[ T_x\, C]   =  1  
 \end{align}
is required to hold for every no-signalling channel $\map C$. 

Consider the problem of distinguishing between a set of no-signalling channels $\{\map C_x\}$ using an indefinite tester.   For every probability distribution $\{\pi_x\}$, the  worst-case probability of error satisfies the bound 
\begin{align}\label{easy}
p^{\rm ind}_{\rm err}  \ge    1  -  \sum_x  \pi_x  \,  \Tr[ T_x  C_x ] \, .
\end{align}
Now, suppose that there exists a constant $\lambda$ and a no-signalling channel $\map C$ such that 
\begin{align}\label{semidefinite}
\lambda \,  C \ge   \pi_x \,  C_x 
\end{align} 
for every $x$. Substituting Equation (\ref{semidefinite})   into Equation (\ref{easy}) 
one obtains the bound 
\begin{align}\label{dualbound}
p^{\rm ind}_{\rm err}  \ge    1  -  \lambda \sum_x    \,  \Tr[ T_x  C ]   =  1-\lambda \, ,
\end{align}
having used the normalization condition (\ref{testernorm}).   The bound (\ref{dualbound}) can be seen as a generalization of the classical Yuen-Kennedy-Lax bound for quantum state discrimination \cite{yuen1975optimum}.   

 We now apply the bound (\ref{dualbound})  to the task of distinguishing between the two channels $\map C_{1, I} =    (\map U \otimes I/d)^{\otimes N}$ and $\map C_{2,I}  =  (  I/d \otimes \map U)^{\otimes N}$. 
   To this purpose, we consider the universal cloning  channel \cite{werner1998optimal}
\begin{align} 
\map C_{\pm}    :  =   \frac {  2}{d^N + 1}   \,  P_{+}  (   \rho \otimes  I^{\otimes N}  )   P_{+}  \, ,
\end{align}
and the universal  {\rm NOT} channel \cite{buvzek1999optimal} 
\begin{align} 
\map C_{\pm}    :  =   \frac {  2}{d^N - 1}   \,  P_{-}  (   \rho \otimes  I^{\otimes N}  )   P_{-}  \, ,
\end{align}
with $P_{\pm}  =  (I  \pm  {\rm SWAP})/2$, and $ {\rm SWAP}$ being the unitary operator that swaps between the even and odd output spaces. 
It is easy to verify that both channels are no-signalling.   Moreover,  we find that the convex combination $\map C  =  p_+  \map C_+  + p_-  \map C_-$ with $p_{\pm} =   \sqrt{   \frac{d^N \pm 1}{2 d^N} }  /  \left(    \sqrt{ \frac{d^N+1}{2 d^N} }  + \sqrt{ \frac{d^N-1}{2 d^N} } \right)
$ satisfies the condition (\ref{semidefinite}) with $\lambda =  \frac 12 \left(\sqrt{ \frac{d^N+1}{2 d^N} }  + \sqrt{ \frac{d^N-1}{2 d^N} } \right)^2$ (see Appendix \ref{app8} for technical details). Hence, the bound (\ref{dualbound})  becomes 
\begin{align}
p^{\rm ind}_{\rm err}  \ge     1- \lambda =  \frac{  1  -  \sqrt{  1    -   \frac 1 {d^{2N}}}}2  \ge  \frac 1{4 d^{2N}} \, .
\end{align}
The above bound implies that the discrimination rate of quantum strategies with indefinite order cannot exceed   $2 \log d$.

\medskip

\section*{Acknowledgments}      We thank the referees of this paper  for valuable comments that led to improvements in the paper. We acknowledge  Robert Spekkens, David Schmidt, Lucien Hardy,  Sergii Strelchuk, and Thomas Gonda for stimulating discussions.   
This work is supported by the National Natural
Science Foundation of China through grant 11675136,
the Croucher Foundation, John Templeton Foundation,
Project 60609, Quantum Causal Structures, the Canadian Institute for Advanced Research (CIFAR), the Hong
Research Grant Council through grants 17300317 and  17300918, and the
Foundational Questions Institute through grant FQXi-RFP3-1325.    This  publication  was  made  possible  through  the
support of a grant from the John Templeton Foundation.
The  opinions  expressed  in  this  publication  are  those  of
the  authors  and  do  not  necessarily  reflect  the  views  of
the John Templeton Foundation. This research was supported in part by Perimeter Institute for Theoretical 
Physics. Research at Perimeter Institute is supported by the Government of 
Canada through the Department of Innovation, Science and Economic Development 
Canada and by the Province of Ontario through the Ministry of Research, 
Innovation and Science.


\appendix

\begin{widetext}
\section{Complementarity relation between tests of the causal structure and tests of the functional dependency between cause and effect.}\label{app1}

Here we provide the proof of the complementarity relation  (7) in the main text. 

\subsection{Bound on the error probability for parallel strategies with no reference system}
\setcounter{equation}{1}
The two causal hypotheses are  that the  quantum channel from $A$ to the composite system $B\otimes C$ is either of the form  
$\map C_{1, U_1}   =    \map U_{1,B}  \otimes {I_C}/d$,  
 or of the form
$\map C_{2, U_2}  =   {I_B}/d \otimes \map U_{2,C}$, 
with  $\map U_1(\cdot)  : =  U_1  \cdot U_1^\dag$, $\map U_2(\cdot)  : =  U_2  \cdot U_2^\dag$.  Here,  $U_1$ and $U_2$ are  unitary operations, unknown to the experimenter but fixed throughout the  $N$ rounds of the experiment.   

Here we consider parallel strategies, where the channel $\map C_{x,  U_x}^{ \otimes N }$  (with $x= 1 $ or $x=2$) is applied in parallel on a multipartite input state, as in the following diagram
\begin{figure}[H]
        	\hspace*{175pt}
        	\includegraphics[width=0.65\textwidth]{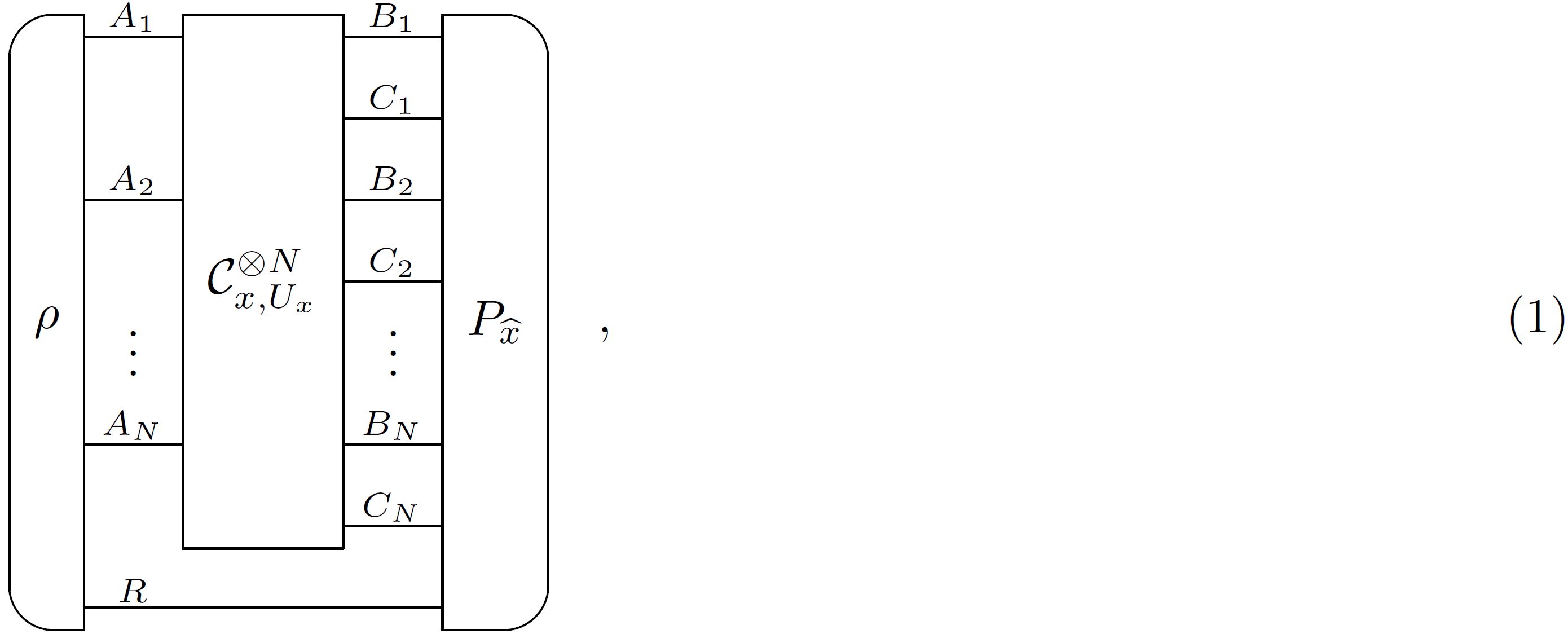}
                	\label{stine}
        \end{figure}
where $R$ is a reference system of fixed dimension. 

The probability to obtain the outcome $\widehat x$ when the channel is $\map C_{x, U_x}$ is equal to 
\begin{align}
p(\widehat x |  x)    =  \Tr \left[  P_{\widehat x} \,  \left(  \map C_{x,  U_x}^{\otimes N} \otimes \map I_R  \right)  (\rho)\right]   \, .
\end{align}
  
For fixed gates $U_1$ and $U_2$, the probability of error  is   
\begin{align}
p_{{\rm err} }  (U_1, U_2)   =    \frac 12    \,  \Tr \left[   P_{1} \,  \left(  \map C_{2,  U_2}^{\otimes N} \otimes \map I_R  \right)  (\rho)  \right]  +  \frac 12\,  \Tr \left[   P_{2} \,  \left(  \map C_{1  ,  U_1}^{\otimes N} \otimes \map I_R  \right)  (\rho) \right]  \, . 
\end{align}
Since $U_1$ and $U_2$ are unknown, we consider the worst-case error probability, namely  
\begin{align}
p_{{\rm err }}^{\rm wc} :  =  \max_{U_1,U_2 \in \grp U}  \, p_{{\rm err} }  (U_1, U_2)    \, ,  
\end{align}
where $\set U$ is a set  of unitary operators.  For example, $\grp U$ can be 
\begin{enumerate}
\item  the group of permutation operators of the form $   U_\pi  =  \sum_{i=1}^{d}  |\pi(i)\>\<i|$, where $\pi$ is an element of the permutation group $S_d$
\item  the group of all unitary operators in dimension $d$. 
\end{enumerate}
In general, we assume that the set $\set U$ is a generalised $N$-design \cite{chiribella2014identifying}, meaning that   (i) $\set U$ is a subset of a group representation $\{  U_g\}_{g\in \grp G}$ for some group $\grp G$, and  (ii) for every operator $A$,  one has the identity
\begin{align}
\frac 1 {|\set U|}\sum_{U \in \set U}  \,  \,    U^{\otimes N}    A     U^{\otimes N \dag}    =  \int_{\grp G}  \d g  \,  U_g^{\otimes N}    A     U_g^{\otimes N \dag}  \, ,      
\end{align}
where $\d g$ denotes the normalized invariant measure over $\grp G$  (for  finite groups,  it is understood that the integral $\int_{\grp G} \d g$ has to be replaced by the sum $\frac 1{|\set G|}  \sum_g$).   

The worst-case error probability is lower bounded by the average error probability
\begin{align}
p_{{\rm err }}^{\rm ave} :  =   \frac 1{ |\set U|^2}   \sum_{U_1,  U_2}\,    p_{{\rm err} }  (U_1, U_2)    \, .  
\end{align}
  By definition, the average  error probability is equal to the error probability in distinguishing between the {\em average channels} 
\begin{align}\label{C1C2}
\map C^{(N)}_1  : =    \frac 1{ |\set U|}\sum_{U_1}\,    \map C_{1,U_1}^{\otimes N} \qquad {\rm and}  \qquad  \map C^{(N)}_2  : =   \frac 1{ |\set U|} \,     \sum_{U_2}    \,  \map C_{2,U_2}^{\otimes N} \,  .
\end{align}

Now, suppose that the experimenter prepares an  $N$-particle state $|\Psi\>  \in \spc H^{\otimes N}$, without using a reference system.  The average error probability has the tight  lower bound 
\begin{align}\label{helstrom}
p^{\rm ave}_{\rm err}      &\ge        \frac {1  -\frac 12  \left \|  \map C_1^{(N)}  (\Psi)  - \map C_2^{(N)}  (\Psi)     \right\|_1 }2 
\end{align}  
achieved by Helstrom's minimum error measurement \cite{helstrom1969quantum}. The distance between the average  output states can be expressed  as 
\begin{align}
\nonumber 
\left \|  \map C_1^{(N)}  (\Psi)  - \map C_2^{(N)}  (\Psi)     \right\|_1    & =  \left \|  \<\Psi\>  \otimes \left(\frac I d\right)^{\otimes N}   - \left(\frac I d\right)^{\otimes N}    \otimes \<\Psi\>      \right\|_1      \qquad \qquad     &\<\Psi\>    :=  \frac 1{ |\set U|}  \,  \sum_{U \in \set U}  \,   U^{\otimes N}   \Psi   U^{\dag \, \otimes N} \\
\nonumber &  =  \frac 1 {d^N}  \,  \sum_{i,j=0}^{d^N-1}  \, |p_i-  p_j|   \qquad     \,  &\<\Psi\>  =  \sum_{i  =0}^{d^N-1}  p_i  \,  |i\>\<i| \\
\nonumber &  =  \frac 1 {d^N}  \,  \sum_{k=1}^{d^N-1}    \left\|  \<\Psi\>   -   S^k  \<\Psi\>  S^{k  \dag  }   \right\|_1  & S :=  \sum_{i=0}^{d^N-1}  \,   |(i+1) \,{\rm mod} \, d^N \>\<   i | \\
\label{mnbv} &  =   \left(  1- \frac 1 {d^N} \right)  \,    \left\|     \<\Psi\>  \otimes \omega   -  \Sigma    \right\|_1  \, ,
\end{align}      
with 
\begin{align}
\omega:  = \frac{  \sum_{k=1}^{d^N  -1}  \,  |k\>\<k|}{d^N-1} \qquad {\rm and}  \qquad  \Sigma  := \frac 1{d^N-1}    \sum_{k=1}^{d^N-1}     S^k  \<\Psi\>  S^{k  \dag  }     \otimes |k\>\<k|   \, .
\end{align}
Now, the pure states 
\begin{align}
\nonumber |\Gamma\>  &:  =  \sum_{i=0}^{d^N-1}  \sum_{k=1}^{d^N-1}  \, \sqrt{\frac{p_i}{d^N-1}}  \,  |i\>\otimes |i\>  \otimes |k\> \otimes |k\>   \\
|\Delta\>   &:  =   \sum_{j=0}^{d^N-1}  \sum_{l=1}^{d^N-1}  \, \sqrt{\frac{p_i}{d^N-1}}  \,  S^k  |j\>\otimes S^k |j\>  \otimes |l\> \otimes |l\>  
\end{align}
are purifications of $\<\Psi\>\otimes \omega$ and $\Sigma$, respectively.  Hence, the monotonicity of the trace distance yields the bound  
 \begin{align}
\nonumber \left \|  \map C_1^{(N)}  (\Psi)  - \map C_2^{(N)}  (\Psi)     \right\|_1      &\le   \left(  1- \frac 1 {d^N} \right)  \,    \left\|    \Gamma   -  \Delta   \right\|_1 \\
\nonumber &  =  \left(  1- \frac 1 {d^N} \right)  \,  2\,  \sqrt{ 1  -  | \<\Gamma  |\Delta  \>|^2}\\
\label{nextnext}  & \le   \left(  1- \frac 1 {d^N} \right)  \,  2\,     \left( 1  -  \frac {| \<\Gamma  |\Delta  \>|^2}2\right)       
  \, .  
\end{align}    
Inserting this bound into Equation  (\ref{helstrom}), we then obtain  
\begin{align}\label{boundbound}
p_{\rm err}^{\rm ave}  \ge  \frac 1 {2d^N}  \, \left[   1  +   \left( d^N  -1   \right) \frac {  |\<\Gamma|\Delta\>|^2}2 \right]
\end{align}

Now, note that we have 
\begin{align}
\nonumber \<\Gamma|\Delta \>   & =  \frac 1 {d^N-1} \, \sum_{i,j=0}^{d^N-1}     \, \sqrt{p_ip_j}\,    ~  \<i| \left(  \sum_{k=1}^{d^N-1}   S^k  |j\>\<j|  S^k \right)  \,  |i\>  \\
\nonumber &   =  \frac 1 {d^N-1} \, \sum_{i,j=0}^{d^N-1}     \, \sqrt{p_ip_j}\,    ~  \<i| \Big(    I-  |j\>\<j|   \Big)  \,  |i\>  \\  
\nonumber &   =  \frac 1 {d^N-1} \, \sum_{i,j=0}^{d^N-1}     \, \sqrt{p_ip_j}\,    ~   (1  - \delta_{ij}) \\
&  =\frac {  \left(\Tr\left[\sqrt{\<\Psi\>}\right]\right)^2 -1}{d^N-1} \end{align}
Hence, the Equation (\ref{boundbound}) yields the bound 
\begin{align}\label{tradeoff1}
p_{\rm err}^{\rm wc}    \ge p_{\rm err}^{\rm ave}  \ge   \frac 1 {2d^N}  \left\{   1  + \frac {\left[  \left(\Tr\left[\sqrt{\<\Psi\>}\right]\right)^2 -1 \right]^2} {  2   (d^N-1)}  \,    \right\} \, .
\end{align}      

It is clear that the minimum of the right-hand-side is obtained when the state $\<\Psi\>$ is pure, in which case, the bound becomes $p_{\rm err}^{\rm wc}  \ge 1/(2d^N)$.

\subsection{Bound on the success probability in the identification of a unitary gate}  
More generally, the bound (\ref{tradeoff1}) can be interpreted as a complementarity relation between the estimation of the causal structure and the estimation of the functional dependence between cause and effect.    

\begin{lemma}
Consider the task of guessing the gate $U  \in \set U$ from the state $|\Psi_U  \>  := U^{\otimes N} |\Psi\>$.   If  $\set U$ is a generalised $N$-design for some group representation  $\{  U \}_{U \in \grp G}$, then the probability of a correct guess satisfies the bound 
\begin{align}\label{symmetry}
p_{\rm guess}^{\set U}  \le    \frac{  \left(  \Tr  \left[  \sqrt{\<\Psi\>}\right]\right)^2}{ |\set  U|}   
\end{align} 
 The bound is attained by the square-root measurement \cite{hausladen1994pretty}, with operators $P_U  =   \< \Psi\>^{-\frac 12}  \,  \Psi_U    \< \Psi\>^{-\frac 12}/|\set U| $.    
\end{lemma} 
\begin{proof} Equation (\ref{symmetry}) follows from  the Yuen-Kennedy-Lax  bound \cite{yuen1975optimum}   $p_{\rm guess}^{\set U}  \le  \Tr[\Lambda]$ where $\Lambda$ is a positive operator satisfying the inequalities $ \Lambda  \ge   \frac 1 {|\set U|} \,   U^{\otimes N}   \Psi U^{\otimes N \dag}$ for all $U \in\set U$.   
Equivalently, one has $U^{\otimes N\dag}  \Lambda  U^{\otimes N}  \le \Psi/|\set U|$ for all $U$, which implies the condition 
\begin{align}
\quad \<\Lambda\> \ge   \frac {\Psi} {|\set U|} \,  , \qquad \<\Lambda\> :  = \frac1 {|\set U|} \sum_U   \,    U^{\otimes N  \dag}  \Lambda  U^{\otimes N}  \, . 
 \end{align} 
Then, the Yuen-Kennedy-Lax  bound  implies the inequality  
\begin{align}
p_{\rm guess}^{\set U}  \le  \Tr[\< \Lambda \>]  \, . 
\end{align}
Since the unitaries $\set U$ form a generalised $N$-design, the operator $\<\Lambda\>$ is invariant under the action of the group representation $\{  U^{\otimes N} \}_{U  \in  \grp G}$.  Moreover, every invariant operator $\Gamma$ can be written as $\<\Lambda\>$ for some suitable $\Lambda$  (in fact, it suffices to take $\Lambda  = \Gamma$).   
Hence, one has the bound 
\begin{align}\label{design}
p_{\rm guess}^{\set U}   \le  \Tr [\Gamma]  \, , \qquad    \forall \Gamma   :  \,  \<\Gamma\>  =  \Gamma  \, ,   \quad          \Gamma    \ge  \frac{\Psi}{|\set U|} \, .    
\end{align} 

In particular, one can take  $\Gamma  =  c\, \sqrt{  \<\Psi\>}$ for some suitable constant $c$. With this choice, the condition $\Gamma \ge \Psi/|\set U|$ is equivalent to 
  \begin{align}
  c  \, I \ge   \frac{ \<\Psi\>^{-\frac 14 }  \Psi \,  \<\Psi\>^{-\frac 14} }{|\set U|}  \, , 
  \end{align} 
  which in turn is equivalent to
  \begin{align}
  c  &  \ge   \frac{\Tr  [\Psi   \< \Psi\>^{-\frac 12}  ] }{|\set U|}  =  \frac{\Tr  [ \<\Psi \>  \< \Psi\>^{-\frac 12} ]}{|\set U|}    = \frac{ \Tr [  \< \Psi \>^{\frac 12}]}{|\set U|} \, .
  \end{align}
  Then, the bound (\ref{design}) becomes $p_{\rm guess}^{\set U}   \le  \Tr [ \sqrt{\< \Psi\>}]^2/|\set U|  $. 
  The bound is attained by the square-root measurement $P_U  =   \< \Psi\>^{-\frac 12}  \,  \Psi_U    \< \Psi\>^{-\frac 12}/|\set U| $, which yields 
  \begin{align}
\nonumber  \Tr  [  P_U  \,  \Psi_U]    & =\frac 1 {|\set U|} \,  \Tr  [  \< \Psi\>^{-\frac 12}  \,  \Psi_U    \< \Psi\>^{-\frac 12}   \Psi_U ]   \\
\nonumber & =  \frac 1 {|\set U|} \,    \Tr  [  \< \Psi\>^{-\frac 12}  \,  \Psi     \< \Psi\>^{-\frac 12}   \Psi  ] \\
\nonumber & =  \frac 1 {|\set U|} \,      \Big|\< \Psi  |     \< \Psi\>^{-\frac 12}  |\Psi\>\Big|^2  \\   
\nonumber & =  \frac 1 {|\set U|} \,      \Big|  \Tr [  \Psi  \< \Psi\>^{-\frac 12}  ]  \Big|^2 \\  
\nonumber 
 & =  \frac 1 {|\set U|} \,      \Big|  \Tr [ \< \Psi \> \< \Psi\>^{-\frac 12}  ]  \Big|^2   \\
 & =  \frac 1 {|\set U|} \,      \Big|  \Tr [ \< \Psi\>^{\frac 12}  ]  \Big|^2   \, ,
  \end{align}
  for every $U$.  
 \end{proof}

Combining the above lemma with Equation  (\ref{tradeoff1})  we obtain the relation 
 \begin{align}\label{tradeoff}
p_{\rm err}^{\rm wc}    \ge  \frac 1{2d^N}    \left\{ 1    +     \frac 1{2(d^N-1)}  \left(  \frac{  p^{\set U}_{\rm guess}  -  \frac 1{|\set U|}  }{  \frac 1  {|\set U|}}   \right)^2     \right\} \ .
\end{align}

\section{Optimal universal  strategy} \label{app2}

Here  we derive the optimal strategy for identifying the causal intermediary when the cause-effect relationship is described by an arbitrary unitary gate.

\subsection{Reduction to the minimisation of the  average probability} 

The problem is to find the strategy    that minimises the worst-case error probability.  Thanks to the symmetry of the problem, the minimisation of the  worst-case error probability  can be reduced to the minimisation of the average  error  probability: 

\begin{lemma}\label{lem:averagepar}
For every fixed reference system $R$ and for every fixed $N$, minimum worst-case error probability in the discrimination of the channels $
\map C_{1,U_1}$  and $\map C_{2,  U_2}$ with $N$ uses is equal to the average error probability
\begin{align}
p_{{\rm err }}^{\rm ave} :  =  \int_{U_1 \in \grp{SU} (d)}  \d U_1 \,  \int_{U_2\in  \grp{SU} (d)  }\d U_2 \,   \, p_{{\rm err} }  (U_1, U_2)    \, ,  
\end{align}
where $\d U$ is the normalised invariant measure. In turn, the average error probability is equal to the minimum error probability in the discrimination of the channels 
\begin{align}\label{C1C2}
\map C^{(N)}_1  : =    \int \d U _1   ~  \map C_{1,U_1}^{\otimes N} \qquad {\rm and}  \qquad  \map C^{(N)}_2  : =    \int \d U_2    ~  \map C_{2,U_2}^{\otimes N} \,  .
\end{align}  
 There exists a state $\rho$ and a measurement $\{  P_1,  P_2\}$ that are optimal for both problems. 
\end{lemma}

\medskip

We omit the proof, which is a simple adaptation of Holevo's argument on the optimality of covariant measurements \cite{holevo2011probabilistic}, see also  \cite{chiribella2011group}.

\subsection{Optimal form of the input states}  
Let us  search for the optimal quantum strategy.  Note that the channels $\map C_1^{(N)}$ and $\map C_2^{(N)}$ satisfy the condition  
\begin{align}\label{twirl}
\map C^{(N)}_x   =    
\map C^{(N)}_x  \circ \map T^{(N)}_{\rm in}  \, , \qquad\forall  x\in\{1,2\} \, .
\end{align}
where $\map T^{(N)}_{\rm in}$ 
is  the twirling channel 
\begin{align}
\map T^{(N)}_{\rm in}  :=  \int  \d W  \,     \map W^{\otimes N} \, .
\end{align}  
Eq. (\ref{twirl}) implies   that the search of the optimal input state  can be restricted to invariant states---{\em i.~e.~}states satisfying the 
condition 
\begin{align}\label{invariantinput}
   \map T_{\rm in}^{(N)}  (\rho)   =  \rho  \, .
\end{align}
 The structure of the invariant states can be made explicit using the Schur-Weyl duality \cite{fulton2013representation}, whereby the tensor product Hilbert space $\spc H^{\otimes N}$ is decomposed as 
\begin{align}\label{schurweyl}
\spc H^{\otimes N}  =  \bigoplus_{\lambda \in  \set Y_{N,d}}  \,   \big(  \spc R_\lambda\otimes \map M_\lambda\big) \, ,  
\end{align} 
 where $\set{Y}_{N,d}$ is the set of  Young diagrams of $N$ boxes arranged in $d$ rows, while  $\spc R_\lambda$ and $\spc M_\lambda$ are representation and multiplicity spaces for the tensor action of  $\grp {SU} (d)$, respectively. Using the Schur-Weyl decomposition, every invariant state on $\spc H^{\otimes N}\otimes \spc H_R$  can be decomposed as 
 \begin{align}\label{invariantinput2}
 \rho   =   \bigoplus_{\lambda}  \,        q_\lambda \,       \left(      \frac {P_\lambda}{d_\lambda}  \otimes  \rho_{\lambda R}   \right) \, ,
 \end{align}
 where $\{q_\lambda\}$ is a probability distribution,  $P_\lambda$ is the identity operator on the representation space $\spc R_\lambda$,  and $\rho_{\lambda R}$ is a density matrix on the Hilbert space $\map M_\lambda\otimes \spc H_R$. 
 
Note that the set of invariant states (\ref{invariantinput2}) is convex. Since the  (average) error probability  is a linear function of $\rho$, the minimisation can be restricted to the extreme points of this convex set.   Hence, we have the following  
\begin{prop}\label{prop:optimal}
Without loss of generality, the optimal input state for a parallel strategy with reference system $R$ can be taken of the form
\begin{align}\label{optimalstatelambda}
\rho  =     \frac {P_{\lambda_0}}{d_{\lambda_0}}  \otimes    \Psi_{\lambda_0 R}     \, ,
\end{align}
where $\lambda_0\in  \set Y_{N,d}$ is a fixed Young diagram and $ \Psi_{\lambda_0 R}  $ is a pure state on $\spc M_{\lambda_0}\otimes \spc H_R$. 
\end{prop}

    \subsection{Error probability for states of the optimal form\label{app:prob}}
  
The problem is to find the input state that makes the output states most distinguishable.  To this purpose, it is convenient to label operators with the corresponding systems and to use the notation ${\bf A}  :=  A_1A_2  \cdots A_N $, ${\bf B}  :=  B_1 B_2 \cdots B_N$,  ${\bf C}  : =  C_1 C_2 \cdots C_N$, and ${\bf R}:  =  R$.

    When applied to an invariant state of the composite system $\bf A   R$, the two channels $\map C^{(N)}_1$ and $\map C^{(N)}_2$ produce the output states 
    \begin{align}\label{aa}
    \Big(  \map C^{(N)}_1\otimes \map I_{\bf R}  \Big)   (\rho_{\bf AR})   =   \rho_{\bf BR}  \otimes  \left(  \frac I d\right)^{\otimes N}_{\bf C}  \qquad {\rm and} \qquad      \Big (\map  C^{(N)}_2 \otimes \map I_{\bf R} \Big)    (\rho_{\bf AR})  =   \left(  \frac I d\right)^{\otimes N}_{\bf B}\otimes   \rho_{\bf CR}   \, ,
    \end{align}
up to a convenient reordering of the Hilbert spaces.

The minimum error probability  in the discrimination of the output states is given by Helstrom's theorem  \cite{helstrom1969quantum}.  Specifically, one has
\begin{align}\label{errorprob}
p_{\rm{err}}   &= \frac{1}{2} \left( 1 - \frac{1}{2} \left\|   \Delta   \right\|_1 \right)  \, , \qquad \Delta    :=  \rho_{\bf BR}  \otimes  \left(  \frac I d\right)^{\otimes N}_{\bf C}   -      \left(  \frac I d\right)^{\otimes N}_{\bf B}\otimes   \rho_{\bf CR}    \, .
\end{align}
In the following, we compute the trace norm  explicitly for input states of the optimal form 
\begin{align}\label{bb}
\rho  =     \frac {P_{\lambda_0}}{d_{\lambda_0}}  \otimes    \Psi_{\lambda_0 R}  \, ,
\end{align}
It is convenient to decompose the identity operator $I^{\otimes N}$ as 
\begin{align}\label{cc}
I^{\otimes N}   =   \bigoplus_{\lambda\in\set Y_{N,d}}  \,  \Big (    P_\lambda  \otimes Q_\lambda \Big) \, ,  
\end{align} 
where  $P_\lambda$ is the identity operator on the representation space $\spc R_\lambda$ and  $Q_\lambda$ is the identity operator on the multiplicity space  $\map M_\lambda$.  In the following, we denote by $m_{\lambda}=  \Tr[Q_\lambda]$ the dimension of $\map M_\lambda$.
 Combining Eqs. (\ref{aa}), (\ref{bb}), and (\ref{cc}), we obtain 
 \begin{align}
\nonumber  \left\|  \Delta    \right\|_1   &=     \frac {d_{\lambda_0}  m_{\lambda_0}}{d^N}    \, 
  \left\|    \frac{ P_{\lambda_0}}{d_{\lambda_0}}   \otimes \frac{ P_{\lambda_0}}{d_{\lambda_0}}    \otimes \left (  \Psi_{\lambda_0 R}     \otimes \frac{Q_{\lambda_0}}{m_{\lambda_0}}     -   \frac{Q_{\lambda_0}}{m_{\lambda_0}} \otimes  \Psi_{\lambda_0 R}            \right) \right \|_1   
  \\
  \nonumber   &\quad  +  2 \sum_{\lambda \not  =  \lambda_0}   \,  \frac {d_{\lambda}  m_{\lambda}}{d^N}    \, 
  \left\|    \frac{ P_{\lambda_0}}{d_{\lambda_0}}   \otimes \frac{ P_{\lambda}}{d_{\lambda}}    \otimes   \Psi_{\lambda_0 R}     \otimes \frac{Q_{\lambda}}{m_{\lambda}}\right\|_1  \\
  &  =       \frac {d_{\lambda_0}  m_{\lambda_0}}{d^N}    \, 
  \left\|       \Psi_{\lambda_0 R}     \otimes \frac{Q_{\lambda_0}}{m_{\lambda_0}}     -   \frac{Q_{\lambda_0}}{m_{\lambda_0}} \otimes  \Psi_{\lambda_0 R}            \right \|_1   +  2 \left( 1  -  \frac{d_{\lambda_0} m_{\lambda_0}}{d^N} \right)  \label{Delta}
     \end{align}

It remains to compute the trace norm in the first summand.  
  To this purpose, it is convenient to define the states
\begin{align}
\left|\Phi^{\pm}_{n} \right\rangle := \frac{  \left|  \Psi_{\lambda_0 R} \right\> \otimes |n \>    \pm  |n\>   \otimes \left |\Psi_{\lambda_0 R}  \right\>}{\gamma^{\pm}_n} \, , \qquad \gamma^{\pm}_{n}:=  {\sqrt{2 (1\pm \<  n|  \rho  |n\>) }}  \label{fam} \ ,
\end{align}
where   $\rho$ is the marginal state of   $\Psi_{\lambda_0 R}$ on the multiplicity space $\spc M_{\lambda_0}$, and  
   $\{  |n\> \, ,   n=1,\dots,   m_{\lambda_0}\}$  are the eigenvectors of $\rho$. 
   With this definition,   the states  
   \begin{align}\{  \left |\Phi^{k}_n\right\> \, ,    k\in  \{+,-\}  \, ,  n\in\{1,\dots, m_{\lambda_0} \}   \}
   \end{align} are mutually orthogonal. For example, one has  
   \begin{align}
 \nonumber  \<   \Phi^+_m  |\Phi^+_n\>    &  =  \frac{  {\rm Re}  [  \<  m  |  \rho  |  n\>]}{  \gamma^{\pm}_m  \gamma^{\pm}_n  }  \\
   &  =  0 \, ,
   \end{align} 
   the second equality coming from the fact that $\rho$ is diagonal in the basis $\{|n\>\}$.  

In terms of the vectors (\ref{fam}), one can rewrite the relevant terms as  
   \begin{align}
        \Psi_{\lambda_0 R}     \otimes \frac{Q_{\lambda_0}}{m_{\lambda_0}}     -   \frac{Q_{\lambda_0}}{m_{\lambda_0}} \otimes  \Psi_{\lambda_0 R}              &   =  \frac 1{ 2 m_{\lambda_0}}   ~\sum_n  \,         \gamma_n^+       \gamma_n^-  \bigg(   \,  \left| \Phi^+_n\right\>\left\<  \Phi^-_n  \right|   +  \left| \Phi^-_n\right\>\left\<  \Phi^+_n\right|   \bigg)  \, .
   \end{align}
Then, the trace norm is
\begin{align}
\nonumber   \left\|     \Psi_{\lambda_0 R}     \otimes \frac{Q_{\lambda_0}}{m_{\lambda_0}}     -   \frac{Q_{\lambda_0}}{m_{\lambda_0}} \otimes  \Psi_{\lambda_0 R}          \right\|_1    &  = \frac 1{ 2 m_{\lambda_0}}   ~\sum_n  \,         \gamma_n^+       \gamma_n^-    \bigg\|  \,  \left| \Phi^+_n\right\>\left\<  \Phi^-_n  \right|   +  \left| \Phi^-_n\right\>\left\<  \Phi^+_n\right|   \bigg\|_1 \\ 
  &=  \frac{2}{m_{\lambda_0}}    \, \sum_n  \,     \sqrt{ 1- \<  n| \rho  |n\>^2} \, .
   \end{align}
The maximum trace norm is reached when the eigenvalues of $\rho$ are all equal.  In that case, one has 
\begin{align}\label{crucial}
  \left\|     \Psi_{\lambda_0 R}     \otimes \frac{Q_{\lambda_0}}{m_{\lambda_0}}     -   \frac{Q_{\lambda_0}}{m_{\lambda_0}} \otimes  \Psi_{\lambda_0 R}          \right\|_1    &   =  \frac{2}{m_{\lambda_0}}     
  \bigg  (  m_{\lambda_0}   -  r   +  r  \sqrt{ 1-r^{-2}}  \bigg)   \, ,
   \end{align}  
     where $r$ is the rank of $\rho$. 
Combining the above equation with Eqs. (\ref{Delta}) and (\ref{errorprob}) we obtain the error probability  
\begin{align}
\label{perrore}
p_{\rm{err}}  &=    \frac{d_{\lambda_0} }{2 d^N}  \,  f(r)   \,  \qquad   f(r): =   r \,\left(    1   -    \sqrt{ 1-r^{-2}}  \,\right) \, .
\end{align}    
 Note that the function $f(r)$ is monotonically decreasing, and therefore the error  probability is minimised by maximising the rank $r$, {\em i.~e.~}by choosing 
 \begin{align}\label{r}
 r=  \min\{  m_{\lambda_0} \, , d_R \} \, ,
 \end{align} 
 where $d_R$ is the dimension of the reference system.   

\subsection{Minimum error probability}
The probability of error is given by Eq.~(\ref{perrore}).   When the reference system has dimension larger than the multiplicity $m_{\lambda_0}$, one has the equality 
\begin{align}
r=  m_{\lambda_0}
\end{align}
and the error probability becomes
\begin{align}\label{errorprobbecomes}
p_{\rm{err}}  &=    \frac{d_{\lambda_0} }{2 d^N}  \,  f(m_{\lambda_0})   \, ,
\end{align}    
with $f$ defined as in Equation (\ref{perrore}). 

The only way to beat the classical scaling $1/d^N$ is to make $f(m_{\lambda_0})$ exponentially small.   Since $f$ is positive and  monotonically decreasing, this means that $m_{\lambda_0}$ must be exponentially large.  
Note that, for large $m_{\lambda_0}$, the probability of error has  the asymptotic expression  
 \begin{align}\label{larger}
p_{\rm{err}}  &=    \frac{d_{\lambda_0} }{4m_{\lambda_0}  d^N }\,  \left[1  +  O\left({m_{\lambda_0}}^{-2}\right)  \right] \, .
\end{align}    
Asymptotically, the problem is reduced to the minimisation of the ratio $d_{\lambda_0}/m_{\lambda_0}$.

 To find the minimum, it is useful to apply the notion of majorisation Young diagrams. Given two diagrams $\lambda$ and $\mu$ of $N$  boxes arranged in $d$ rows, we say that {\em $\lambda$ majorises $\mu$} if 
 \begin{align}
 \sum_{i=1}^s  \, \lambda_i  \ge \sum_{i=1}^s  \,\mu_i  \qquad \forall   s\in  \{1,\dots, d\} \, , 
 \end{align}     
 where $\lambda_i$ ($\mu_i$) is the length of the $i$-th row of the diagram $\lambda$  ($\mu$).

 \begin{lemma}\label{lem:majorisation}
 If $\lambda$ majorises $\mu$, then $d_\lambda/m_\lambda  \ge   d_\mu/m_\mu$. 
 \end{lemma}
 \begin{proof}    For a generic Young diagram $\lambda \in  \set Y_{N+1,d}$, one has 
\begin{align}
d_{\lambda}=\frac{\prod_{(i,j) \in \lambda} d - i+j}{\prod_{(i,j)\in \lambda} \text{hook}(i,j)} \qquad {\rm and}  \qquad m_{\lambda}= \frac{N!}{\prod_{(i,j)\in \lambda} \text{hook}(i,j)} \ , \label{dimensions}
\end{align}
Here the pair $(i,j)$ labels a box in the diagram, with the indices $i$ and $j$ labelling the  row and  the column, respectively.   $ \text{hook}(i,j)$ denotes the  length of the hook consisting of boxes to the right and to the bottom of  the box $(i,j)$. 
Using the above expressions, the dimension/multiplicity ratio reads
\begin{align}
\nonumber \frac{d_{\lambda}}{m_{\lambda}}  &=\frac{\prod_{(i,j) \in \lambda} d   - i+j}{N!}  \\
  &  =    \,   \frac 1 { N!} \,   \prod_{i=1}^d    \frac{     (d-i+ \lambda_i)!      }{(d-i)!} \ .   \label{ratio}
\end{align}
 Now, since $\lambda$ majorises $\mu$, one has the bounds
  \begin{align}
\nonumber   \frac{     (d-1+ \lambda_1)!      }{(d-1)!}    &\ge   \frac{     (d-1+ \mu_1)!      }{(d-1)!}  \,   (d+\mu_1)^{\lambda_1  -  \mu_1}   \\ 
\nonumber   \frac{     (d-1+ \lambda_1)!      }{(d-1)!}   \,    \frac{     (d-2+ \lambda_2)!      }{(d-2)!}    &   \ge       \frac{     (d-1+ \mu_1)!      }{(d-1)!}   \,    \frac{     (d-2+ \mu_2)!      }{(d-2)!}       \,   (d-1 +  \mu_2)^{\lambda_1  +  \lambda_2  -  \mu_1-\mu_2}   \,   \\
\nonumber    & ~\,  \vdots   \\
\prod_{i=1}^s    \frac{     (d-i+ \lambda_i)!      }{(d-i)!}  & \ge  \prod_{i=1}^s    \frac{     (d-i+ \mu_i)!      }{(d-i)!}  \,    (d-s+1  +  \mu_s)^{\sum_{i=1}^s   (  \lambda_i-  \mu_i)}   \qquad \forall s\in  \{1,\dots,d\}\, .
 \end{align}
 Choosing $s=d$ and recalling Eq. (\ref{ratio}), one finally obtains  $d_\lambda/m_\lambda \ge d_\mu/m_\mu$. 
 \end{proof}  
 
 \begin{prop}
Define $t  := N-d   \lfloor N/d  \rfloor $. Then, the ratio $d_{\lambda}/m_{\lambda}$ is 
\begin{enumerate}
\item minimum when $\lambda$ is the  Young diagram with $t$ rows of length $\lceil N/d  \rceil$ and $d-t$ rows of length $\lfloor N/d  \rfloor$
\item   maximum when $\lambda $ is the Young diagram with one row or length $N$. 
 \end{enumerate}
 \end{prop}

\begin{proof}  The Young diagram $\lambda_0  =   (  \underbrace{\lceil N/d  \rceil  ,   \dots,  \lceil N/d  \rceil}_{t~{\rm times}} ,  \underbrace{\lfloor N/d  \rfloor  ,   \dots,  \lfloor N/d  \rfloor}_{d-t~{\rm times}} )$ is majorised by any other Young diagram in $\set Y_{N,d}$. Hence,  $\lambda_0$ minimises the ratio $d_\lambda/m_{\lambda}$  (by Lemma \ref{lem:majorisation}).    Similarly, the Young diagram  $\lambda_0  =  (N,  \underbrace{0,\dots, 0}_{d-1~{\rm times}})$ majorises every other young diagram and therefore it maximises the ratio $d_\lambda/m_\lambda$.  
\end{proof} 

\medskip 

Summarizing, we showed that 
\begin{enumerate}
\item when $N$ is a multiple of $d$, the optimal Young diagram corresponds to the trivial representation of $\grp {SU} (d)$
\item when $N$ is not a multiple of $d$, the optimal  Young diagram corresponds to the totally antisymmetric representation acting on  $N-d   \lfloor N/d  \rfloor $ particles.  
\item asymptotically, the symmetric subspace is the worst possible choice, leading to the classical rate $R_{\rm C}  = \log d$. 
\end{enumerate}
 
In conclusion, we proved the following 
\begin{prop}\label{prop:optimaltrivial}
When $N$ is a multiple of  $d$, the optimal input state is  $ P_{\lambda_0}/d_{\lambda_0}  \otimes  |\Psi\>\<\Psi|_{\lambda_0 \bf R }$, where $\lambda_0$ is he trivial representation of $\grp{SU} (d)$ in the $N$-fold tensor product  $U^{\otimes N}$, $d_{R}  \ge  m_{\lambda_0}$, and $|\Psi\>_{\lambda_0 \bf R }   \in  \spc M_{\lambda_0} \otimes \spc H_R$  is a maximally entangled state. 
 \end{prop}
 
 Since the trivial representation is one-dimensional, the error probability (\ref{larger}) takes the form
\begin{align}\label{largerthan}
p_{\rm{err}}  &=    \frac{1}{4m_{\lambda_0}  d^N }\,  \left[1  +  O\left({m_{\lambda_0}}^{-2}\right)  \right] \, .
\end{align}    
Moreover, the trivial representation of $\grp{SU} (d)$ corresponds to the Young diagram with $ d$ rows, each of length $N/d$.  
Hence, its multiplicity is given by
\begin{align}
m_{\lambda_0}&= \frac{N!}{\prod_{i=1}^d \frac{\left( \frac{N}{d}+d-i \right)!}{(d-i)!}} \ .
\end{align}
For fixed $d$, the Stirling approximation yields the expression 
\begin{align}\label{multiexpression}
m_{\lambda_0}   =    d^N  \, \left[ \frac{ d^{\frac{d^2}2}    \, e^{\frac{d(d-1)}2}\,  \prod_{i=1}^d (d-i)!} {(2  \pi)^{\frac{d-1}2} \,  N^{\frac{d^2-1}2}} \right] \,  c(N) \, ,
\end{align}  
where $c(N)$ is a function tending to 1 in the large $N$ limit. Taking the logarithm on both sides, one obtains 
\begin{align}
\log m_{\lambda_0}   =    N  \log d  +   O  (  \log N)  \, .   
\end{align}
Inserting this value into the expression of the error probability (\ref{larger}),  we obtain the rate
\begin{align}
\nonumber R  & =  \lim_{N\to \infty}  - \frac{  \log p_{\rm{err}}}{N}   \\
 \nonumber &=    \lim_{N\to \infty}   \frac{ \log  (4 m_{\lambda_0}  d^N)} {N}   \\
  &  =2 \log d  \, .
\end{align}    
\subsection{Quantum superposition of equivalent setups}

Here we   prove that the optimal state can be realized as a coherent superposition of equivalent setups, where the $N$ input variables are divided in groups of $d$, and all the variables in the same group are initialized in the $\grp {SU} (d)$ singlet state.

\begin{prop}  For $N$  multiple of $d$, consider the state 
\begin{align}
|\Psi\>_{{\bf A  R}}  =   \frac 1 {\sqrt{ G_{N,d}}}  \,    \sum_i \,  \left(  |S\>^{\otimes {N/d}}_{{\bf A}} \right)_i \otimes |i\>_{\bf R} \, ,
\end{align}
    where $\{|i\>_{\bf  R}\}_{i=1}^{G_{N,d}}$ is an orthonormal basis for the reference system, indexed by  the possible ways to group $N$  objects into groups of $d$,  and $\left( |S\>^{\otimes {N/d}} \right)_i$ is the product of $N/d$ singlet states, distributed according to the grouping $i$.     Then,  
    \begin{enumerate}
    \item the state  $|\Psi\>_{{\bf A  R}}$ is optimal for the identification of the causal intermediary
    \item the number  $r$ of linearly independent vectors of the form $\left(  |S\>^{\otimes {N/d}}_{{\bf A}} \right)_i$ satisfies the equality 
    \begin{align}
r     =    d^N  \,   \left[ \frac{ d^{\frac{d^2}2}    \, e^{\frac{d(d-1)}2}\,  \prod_{i=1}^d (d-i)!} {(2  \pi)^{\frac{d-1}2} \,  N^{\frac{d^2-1}2}} \right] \, c(N) \, ,
\end{align}  
where $c(N)$ is a function tending to 1  in the large $N$ limit. 
     \end{enumerate}
      \end{prop}
  
 \begin{proof}  By definition, $|\Psi\>_{\bf A R}$  is invariant under the $n$-fold action of $\grp {SU} (d)$ on system $\bf A$, meaning that the corresponding density matrix has the optimal form  $  |\Psi\>\<\Psi|_{\bf AR}  =  P_{\lambda_0}/d_{\lambda_0}  \otimes  |\Psi\>\<\Psi|_{\lambda_0 \bf R }$, where $\lambda_0$ is the trivial representation of $\grp {SU} (d)$.  In fact, since the trivial representation is one-dimensional, we may equivalently write $  |\Psi\>\<\Psi|_{\bf AR}  \equiv  |\Psi\>\<\Psi|_{\lambda_0 \bf R }$.  
  
 Now,  the marginal  state  
  \begin{align}
  \nonumber  \rho_{\bf A}    &:  = \Tr_{\bf R} \left[ |\Psi\>\<\Psi|_{\bf AR}  \right]\\
 \label{nanna}
    &  = \frac 1 {G_{N,d}}     \sum_i \,   \left(  |S\>\<S|^{\otimes {N/d}}_{{\bf A}} \right)_i  
    \end{align}  
     is invariant under permutations.  Hence, the Schur lemma implies the relation 
     \begin{align}\label{nonno}
\rho_{\bf A}   =  \frac {Q_{\lambda_0}}{m_{\lambda_0}} \, .
     \end{align}
Since $|\Psi\>_{\bf AR}$ is a purification of $\rho_{\bf A}$, we conclude that $|\Psi\>_{\bf AR}$ is a maximally entangled state between $R$ and the multiplicity system $M_{\lambda_0}$.  Hence, $|\Psi\>_{\bf AR}$ coincides with the optimal input state of Proposition \ref{prop:optimaltrivial}. 

Moreover, comparing Equations (\ref{nanna}) and (\ref{nonno}) we obtain that  the rank of $\rho_{\bf A}$ is equal to the multiplicity $m_{\lambda_0}$. Since the rank of $\rho_{\bf A}$ is the number of linearly independent vectors of the form $ \left(  |S\>^{\otimes {N/d}}_{{\bf A}} \right)_i $, we conclude that the number of such vectors is $m_{\lambda_0}$.   Finally, $m_{\lambda_0}$ can be expressed as in Equation (\ref{multiexpression}). 
\end{proof}

\section{Optimal classical strategy for  $k$ causal hypotheses \label{app3}}
 Here we provide the optimal classical strategy for the case  where exactly one out of $k$ possible variables $B_1,B_2,\dots,B_k$ is the causal intermediary  of $A$. The result is stated in the following
 \begin{lemma}
The minimum error probability
 in the identification of the causal intermediary among $k\ge 2$ alternatives is
 \begin{align}\label{classicalp} 
\nonumber  p^{\rm C}_{\rm err}     =   \frac{k-1}{2d^{N-1}}     +    O \left(  \frac 1 {d^{2N}}\right)     \, . 
 \end{align} 
 \end{lemma}
 \begin{proof} Suppose that the $i$-th output variable is not the causal intermediary.  The probability that it takes values compatible with a permutation is $P(d,v)/d^{N}$, where $v$ is the number of distinct values of $A$ probed in the experiment and $P(d,v)=  d!/(d-v)!$ is the number of injective functions from a $v$-element set to a $d$-element set.  

 Hence, the probability that the $i$-th variable---and \emph{only} the $i$-th variable---is confusable with the true causal intermediary is    
\begin{align}
p_i    =     \frac{P(d,v)}{d^{N}}  \, \left[1- \frac{P(d,v)}{d^{N}}  \right]^{k-2} \, .
\end{align}
Similarly, the probability that that variables $i_1$, $i_2,  \dots,  i_t$ (and {\em only}  variables $i_1$, $i_2,  \dots,  i_t$) are confusable   with the true causal intermediary is 
\begin{align}p_{i_1i_2\dots i_t}  = \left[ \frac {P(d,v)}{d^{N}}\right]^t \,     \left[1-  \frac{P(d,v)} {d^{N}}\right]^{k-t-1} \,.
\end{align}  
When this situation arises, one has to resort to a random guess, with probability of error $t/(t+1)$. In total, the probability of error is equal to 
 \begin{align}
\nonumber  p^{\rm C}_{\rm err}     & =  \sum_{ t=1}^{k-1}   \,      \frac{  t}{t+1}   \,        \begin{pmatrix}    k-1  \\  t  \end{pmatrix}    \,  \left[\frac{P(d,v)}{d^{N}} \right]^t   \left[1- \frac{P(d,v)}{d^{N}}  \right]^{k-t-1}  \\
  &   =   \frac{(k-1)  \,  P(d,v)}{2d^{N}}     +    O \left(  \frac 1 {d^{2N}}\right)     \, . 
 \end{align}     
 Since the coefficient $P(d,v)$ is minimum when $v=1$, the optimal strategy is to initialize all input variables in the same value, thus obtaining probability of error $  p^{\rm C}_{\rm err}   = \frac{k-1}{2d^{N-1}}     +    O \left(  \frac 1 {d^{2N}}\right)     $.  
 \end{proof}

 \section{Optimal  quantum strategy  for $k$ hypotheses without reference system \label{app4}}

Here we provide the best strategy among all quantum strategies that do not use a reference system.  \begin{lemma}
 The best quantum strategy without reference system is  to divide the $N$ input variables into $N/d$ groups of $d$ elements each and, within each group, to prepare the singlet state
\begin{align}   
|S_d\>    =  \frac 1 {\sqrt d!}    \sum_{k_1,k_2, \cdots,  k_d}  \,  \epsilon_{k_1k_2\dots k_d}   \,   |k_1\>  |k_2\>  \cdots |k_d\> \,
\end{align} 
where $\epsilon_{k_1k_2\dots k_d}$ is the totally antisymmetric tensor and the sum ranges over all vectors in the computational basis. The corresponding error probability  is
\begin{align}\label{perrqck}
p_{\rm err}^{\rm QC}=\frac {k-1}{2d^N}  + O \left(  \frac 1 {d^{2N}}\right)   \, .
\end{align}
 \end{lemma}
\begin{proof} Let us denote by $x$ the ``true causal intermediary", namely the quantum system $B_x$ whose state depends  on the state of $A$, and by $\map C_{x,U}$ the channel  defined by the relation
\begin{align}
\map C_{x,U}  (\rho)   =     \big[\map U (\rho)\big]_x  \otimes \left(\frac{I}d \right)^{\otimes (k-1)}_{\overline x} \, ,
\end{align}  
where the subscript $x$ indicates that the operator $\map U(\rho)$ acts on the Hilbert space of system $B_x$ and the subscript $\overline x$  indicates that the operator acts on the  Hilbert space of the remaining $k-1$ systems. 

By the same arguments used in Lemma \ref{lem:averagepar}, the discrimination of the causal hypotheses can be reduced  to the discrimination of the channels 
\begin{align}
\map C_x^{(N)}  = \int \d U \,   \map C_{x,  U}^{\otimes N} \, ,  \qquad x\in  \{1,\dots, k\} \, .
\end{align}
Again, one can show that, for every reference system $R$,  the optimal state can be chosen of the form 
\begin{align}
\rho  =       \frac{  P_{\lambda_0}}{d_{\lambda_0}} \otimes    \Psi_{\lambda_0  R} \, ,
\end{align}
where $P_{\lambda_0}$ is the projector on the $\grp {SU} (d)$ representation space with Young diagram $\lambda_0$, $d_{\lambda_0}  =  \Tr[  P_{\lambda_0}]$, and $\Psi_{\lambda_0 R}$ is a pure state of the composite system $\spc M_{\lambda_0} \otimes \spc H_R$,  $\spc M_{\lambda_0}$  being the   $\grp {SU(d)}$ multiplicity space associated to $\lambda_0$.   

Here we consider the case where the reference system $R$ is trivial.  In this case, the problem is to distinguish among the states 
\begin{align}
\rho_x  :=         \left(     \frac{  P_{\lambda_0}}{d_{\lambda_0}} \otimes    \Psi_{\lambda_0} \right)_x\otimes \left(\frac  I d  \right)^{\otimes   N\,  (k-1)}_{\overline x}  \qquad x\in  \{ 1,\dots,  k\}\, . 
\end{align}
Using the Yuen-Kennedy-Lax formula \cite{yuen1975optimum},  the  maximum  success probability in distinguishing among these states is 
\begin{align*}
p_{\rm succ}    &  =    \min  \Big \{   \Tr [\Gamma]  ~|~ \Gamma  \ge \frac 1  k  \,\rho_x  \, ,  \quad \forall x \in  \{1,\dots, k\}  \Big\}   \, .   \end{align*}

Note that the states $\{\rho_x\,,  k=1,\dots, k\}$ commute.  Hence,   they can be diagonalized in the same basis and the operator $\Gamma$ can be chosen to be diagonal in that basis without loss of generality.  With a similar argument, one can restrict the search for the optimal $\Gamma$ over  the operators of the form
\begin{align}
\Gamma=\bigoplus_{\lambda_1,\lambda_2,\dots,\lambda_k}  \,  P_{\lambda_1}  \otimes P_{\lambda_2}  \otimes \cdots \otimes P_{\lambda_k}  \otimes   \Gamma_{\lambda_1,\dots, \lambda_k} \, , 
\end{align}
where   $\Gamma_{\lambda_1,\dots, \lambda_k}$ is an operator acting on the tensor product  space $\spc M_{\lambda_1} \otimes \spc M_{\lambda_2}\otimes \cdots \otimes \spc M_{\lambda_k}$.  Note that the operators $\Gamma_{\lambda_1,\dots,  \lambda_k}$ can be set to zero for all $k$-tuples  $(\lambda_1,\dots,  \lambda_k)$ such that $\lambda_i  \not  =  \lambda_0$ for every $i\in  \{1,\dots,  k\}$.  
   Now, suppose that $\lambda_i  = \lambda_0 $ and $\lambda_j  \not = 0$ for the remaining $j\not =  i$.   In this case, we must have 
     \begin{align}
  \Gamma_{\lambda_1, \dots, \lambda_{i-1}  \lambda_0  \lambda_{i+1}  \dots    \lambda_k}  \,     \ge      \frac1 {k d_{\lambda_0}    d^{N(k-1)} }  \,   \,  Q_{\lambda_1} \otimes \cdots \otimes Q_{\lambda_{i-1}}  \otimes \Psi_{\lambda_0}  \otimes   Q_{\lambda_{i+1}}  \otimes \cdots \otimes Q_{\lambda_k} \, ,
  \end{align} 
where $Q_{\lambda}$ is the identity operator on the multiplicity space $\spc M_\lambda$.  
Taking the trace on both sides, we obtain the relation  
   \begin{align}
   \Tr \left[   \Gamma_{\lambda_1, \dots,   \lambda_{i-1}  \lambda_0  \lambda_{i+1}  \dots      \lambda_k } \right] \,     \ge      \frac1 {k d_{\lambda_0}    d^{N(k-1)} } \,     m_{\lambda_1}   \, \dots  m_{\lambda_{i-1} }  \,m_{\lambda_{i+1}}  \dots  m_{\lambda_k} \, .   \end{align}

Similar bounds can be found for the operators $\Gamma_{\lambda_1,\dots, \lambda_k}$ where two or more indices are equal to $\lambda_0$.  For example, consider  the terms where $\lambda_i  =  \lambda_j  =  \lambda_0$, while $\lambda_l  \not  =  0$ for the remaining values of $l$.  In this case, we have the conditions
\begin{align}
\label{uno} \Gamma_{\lambda_1, \dots,      \lambda_k}    &  \ge        \frac1 {k d_{\lambda_0}    d^{N(k-1)} } \,    \Big(     \Psi_{\lambda_0}  \otimes  Q_{\lambda_0} \Big)_{ij} \otimes \Big ( Q_{\lambda}\Big )_{\overline{ij}}  \\
   \label{due}\Gamma_{\lambda_1, \dots,      \lambda_k}    &  \ge        \frac1 {k d_{\lambda_0}    d^{N(k-1)} } \,    \Big(      Q_{\lambda_0} \otimes  \Psi_{\lambda_0}  \Big)_{ij} \otimes \Big ( Q_{\lambda}\Big )_{\overline{ij}}   \, ,
\end{align}   
where we introduced the shorthand notation 
\begin{align} \Big ( Q_{\lambda}\Big )_{\overline{ij}}  :  =  Q_{\lambda_1}  \otimes  \cdots \otimes Q_{\lambda_{i-1}}  \otimes  Q_{\lambda_{i+1}}  \otimes \cdots \otimes Q_{\lambda_{j-1}} \otimes Q_{\lambda_{j+1}} \otimes  \cdots \otimes Q_{\lambda_k} \, .
\end{align}
We now combine conditions (\ref{uno}) and (\ref{due}) can be combined into a single condition. To this purpose, we expand $Q_{\lambda_0}$ as
\begin{align*}
Q_{\lambda_0}=\Psi_{\lambda_0} + \Psi^\perp_{\lambda_0} \, ,
\end{align*}
which allows for rewriting (\ref{uno}) and (\ref{due}) as
\begin{align}
\Gamma_{\lambda_1, \dots,      \lambda_k}    &  \ge        \frac1 {k d_{\lambda_0}    d^{N(k-1)} } \,    \Big(     \Psi_{\lambda_0}  \otimes  \Psi_{\lambda_0} + \Psi_{\lambda_0}  \otimes \Psi^\perp_{\lambda_0} \Big)_{ij} \otimes \Big ( Q_{\lambda}\Big )_{\overline{ij}}  \\
  \Gamma_{\lambda_1, \dots,      \lambda_k}    &  \ge        \frac1 {k d_{\lambda_0}    d^{N(k-1)} } \,    \Big(      \Psi_{\lambda_0}  \otimes  \Psi_{\lambda_0} + \Psi^\perp_{\lambda_0} \otimes \Psi_{\lambda_0}    \Big)_{ij} \otimes \Big ( Q_{\lambda}\Big )_{\overline{ij}}   \, .
\end{align} 
Now, since $\Psi_{\lambda_0}  \otimes \Psi^\perp_{\lambda_0}$ and $\Psi^\perp_{\lambda_0} \otimes \Psi_{\lambda_0}$ are orthogonal vectors, it is also true that
\begin{align*}
\Gamma_{\lambda_1, \dots,      \lambda_k}    &  \ge        \frac1 {k d_{\lambda_0}    d^{N(k-1)} } \,    \Big(     \Psi_{\lambda_0}  \otimes  \Psi_{\lambda_0} + \Psi_{\lambda_0}  \otimes \Psi^\perp_{\lambda_0} + \Psi^\perp_{\lambda_0} \otimes \Psi_{\lambda_0} \Big)_{ij} \otimes \Big ( Q_{\lambda}\Big )_{\overline{ij}} \, ,
\end{align*}
which can be rewritten as  
\begin{align}
 \Gamma_{\lambda_1, \dots,      \lambda_k}      \ge        \frac1 {k d_{\lambda_0}    d^{N(k-1)} } \,    \Big(   Q_{\lambda_0}\otimes   Q_{\lambda_0}  -  \Psi^\perp_{\lambda_0}  \otimes  \Psi^\perp_{\lambda_0} \Big)_{ij} \otimes \Big ( Q_{\lambda}\Big )_{\overline{ij}}  \, .
\end{align}
Tracing on both sides, one obtains 
\begin{align}
\Tr \left[  \Gamma_{\lambda_1, \dots,      \lambda_k}    \right]  \ge        \frac1 {k d_{\lambda_0}    d^{N(k-1)} } \,        ( 2m_{\lambda_0}-1 )    \,  \left(\prod_{l\not  =  i,j}  \, m_{\lambda_{l}}\right) \, .
\end{align}

Likewise, a term with $\lambda_{i_1}  =  \lambda_{i_2}  =  \dots  =  \lambda_{i_t}  =  \lambda_0$ and all the remaining $\lambda_l$ different from $\lambda_0$ will satisfy the condition
\begin{align}
 \Gamma_{\lambda_1, \dots,      \lambda_k}      \ge        \frac1 {k d_{\lambda_0}    d^{N(k-1)} } \,    \Big(   Q_{\lambda_0}^{\otimes t}  -  \Psi^{\perp   \otimes t}_{\lambda_0}   \Big)_{i_1 \dots  i_t} \otimes \Big ( Q_{\lambda}\Big )_{\overline{i_1\dots i_t}}  \, ,
 \end{align}
leading to the inequality  
\begin{align}
\Tr \left[  \Gamma_{\lambda_1, \dots,      \lambda_k} \right] \ge  \frac1 {k d_{\lambda_0}    d^{N(k-1)} } \,    \left[  m_{\lambda_0}^t   -  (m_{\lambda_0}  -1)^t\right]     \, \prod_{l\not  =  i_1,\dots, i_t} \,     m_{\lambda_l} \, .   
\end{align}
Note that one can choose the operator $\Gamma$ in such a way that the equality holds in all bounds. 
With this choice, the probability of success is  
\begin{align}
\nonumber p_{\rm succ}  & = \sum_{\lambda_1,\dots,  \lambda_k}      d_{\lambda_1}  \dots  d_{\lambda_k}   \,  \Tr[\Gamma_{\lambda_1,\dots \,  \lambda_k}] \\
 \nonumber  &  =    \sum_{t=1}^{k}    \begin{pmatrix}   k \\t  \end{pmatrix}     \,        \frac{\left(d_{\lambda_0}m_{\lambda_0}  \right)^t} {k d_{\lambda_0}     d^{N(k-1)} } \,    \left[  1   -  \left(1 -\frac 1 {m_{\lambda_0}}  \right)^t\right]     \,    \left(d^{N}  - d_{\lambda_0}  m_{\lambda_0}\right)^{k-t}  \\
   &  =     \frac{  d^N } {k d_{\lambda_0}    } \, \sum_{t=1}^{k}    \begin{pmatrix}   k \\t  \end{pmatrix}     \,           p_{\lambda_0}^t \,   (1-p_{\lambda_0})^{k-t} \left[  1   -  \left(1 -\frac 1{m_{\lambda_0}}\right)^t\right]   \, ,  
\end{align}
having defined the Schur-Weyl measure $p_{\lambda} :  =    d_\lambda  m_\lambda/ d^N$. 

Expanding the term in square brackets, we obtain  
\begin{align}
\nonumber p_{\rm succ}  & =    \frac{  d^N } {k d_{\lambda_0}    }  \sum_{t=1}^{k}    \begin{pmatrix}   k \\t  \end{pmatrix}     \,     \,    p_{\lambda_0}^t \,   (1-p_{\lambda_0})^{k-t}  \,      \left[   \sum_{s=1}^t    \,  \begin{pmatrix}   t \\s  \end{pmatrix}   \,  \frac{(-1)^{s+1}}{  m_{\lambda_0}^s}  \right] \\
\nonumber &    =         \frac{  d^N } {k d_{\lambda_0}    }  \,  \sum_{s=1}^k\,       \,  \frac{(-1)^{s+1}}{  m_{\lambda_0}^s}    \,  \left[   \sum_{t=s}^{k}                \begin{pmatrix}    k \\t  \end{pmatrix}     \,   \begin{pmatrix}   t \\s  \end{pmatrix}  \,    p_{\lambda_0}^t \,   (1-p_{\lambda_0})^{k-t} \right]  \\
\nonumber &    =     \frac{  d^N } {k d_{\lambda_0}    } \,     \sum_{s=1}^k\,        \frac{(-1)^{s+1}  \,   p_{\lambda_0}^s }{  m_{\lambda_0}^s}      \begin{pmatrix}    k \\s  \end{pmatrix}   \,  \\
\nonumber &    =     \frac{  d^N } {k d_{\lambda_0}    } \,   \left[ 1  -   \left(  1  -\frac{p_{\lambda_0}}{m_{\lambda_0}} \right)^k \right]   \\
&  =  1   -  \frac{(k-1) d_{\lambda_0}}{2 d^N}   +   O\left[ \left(    \frac {d_{\lambda_0}}{d^N} \right)^2\right] \, .
\end{align}\
Hence, the error probability  is 
\begin{align}
p_{\rm err}  =    \frac{(k-1) d_{\lambda_0}}{2 d^N}   +   O\left[ \left(    \frac {d_{\lambda_0}}{d^N} \right)^2\right]  \, .
\end{align} 
Again, the optimal choice for $N$ multiple of $d$ is to pick $\lambda_0$ to be the trivial representation of $\grp {SU} (d)$, in which case the  error probability is   
\begin{align}
p_{\rm err}  =    \frac{(k-1)}{2 d^N}   +   O \left(    \frac {1}{d^{2N}} \right)  \, .
\end{align}   
Note that, however, the choice of representation $\lambda_0$ does not affect the asymptotic rate: indeed, for every $\lambda_0$ we have 
\begin{align}
\nonumber R   & =    - \liminf_{N\to \infty}  \frac {  \log p_{\rm err}}{N}   \\
\nonumber  &  =     \log d     -   \liminf_{N\to \infty}  \frac {  \log  [(k-1) \, d_{\lambda_0}/2 ]}{N}  \\
 \nonumber &  =  \log d  \\
  &  \equiv R_{\rm C} \, .
\end{align}
Note also that the rate is independent of the number of hypotheses, as in the case of the Chernoff bound for quantum states \cite{li2014second}.  \end{proof}

 \section{Optimal quantum strategy for $k$ causal hypotheses with arbitrary reference system \label{app5}}  
Here we provide  the optimal  quantum strategy using a reference system.  We will prove the following lemma:
\begin{lemma} 
The optimal input state is 
\begin{align}
|\rho \> =    \frac 1 {\sqrt {G_{N,d}}}  \,  \sum_{i=1}^{G_{N,d}}  \,     \left( |S_d\>^{\otimes N/d}\right)_i \otimes |i\>  \, ,
\end{align}   
where $i$ labels the different ways to divide $N$ identical objects into groups of $d $ elements,  $G_{N,d}=\frac{N!}{(d!)^{N/d} (N/d)!}$ is the total number of such ways,  $ \left( |S_d\>^{\otimes N/d}\right)_i$ is the product of $N/d$  singlet states arranged according to the configuration  $i$,  and $\{  |i\>  \, , \,  i=  1,\dots,  G_{N,d}\}$ are orthogonal states of the reference system, chosen to be of dimension equal to or larger than $G_{N,d}$. The corresponding error probability is upper bounded as
\begin{align}
 p^{\rm Q}_{\rm{err}}  (r)  &\leq \frac{k-1}{2d^N m(N,d)}   \label{quantumk}
\end{align}
 where   $m(N,d)$ is the dimension of the multiplicity space of the trivial representation, given by (for $N/d$ being an integer)
 \begin{align}\label{asymptoticmulti}
 m(N,d)    =    d^N  \,   \left[ \frac{ d^{\frac{d^2}2}    \, e^{\frac{d(d-1)}2}\,  \prod_{i=1}^d (d-i)!} {(2  \pi)^{\frac{d-1}2} \,  N^{\frac{d^2-1}2}} \right]  \,  c(N)  \, ,
  \end{align}
  with $\lim_{N\to \infty} c(N)  = 1$. 
\end{lemma}
The proof consists of four  steps:  
\medskip

{\em Step 1:  reduction to the permutation register.}  We apply  $N$  uses of the channel $  {\map{C}}_x$  to a state of the optimal form (\ref{optimalstatelambda}), where the pure state $|\Psi_{\lambda_0}\>$ is set to be  the maximally entangled state $|\Phi_{\lambda_0}\>  = \sum_{i=1}^{m_{\lambda_0}}   |i\>\otimes |i\> /\sqrt{m_{\lambda_0}}$. The output state is   
\begin{align}
\rho_x^{\rm out}   =  \left(   \frac{P_{\lambda_0}}{d_{\lambda_0}} \otimes   \Phi_{\lambda_0} \right)_x  \otimes  \left(   \frac  I   d\right)^{\otimes N  (k-1)}_{\overline x}  \, ,    
  \end{align}
 where the subscript  $x$  indicates  that the corresponding operator acts on the   $N$ Hilbert spaces with label  $x$  (and on the  reference), while the subscript $\overline x$ indicates that the corresponding operator acts on  all systems except those with label  $x$.   
 
Breaking down the identity operator as $I   =     (P_{\lambda_0}\otimes Q_{\lambda_0})   \oplus ( I-   P_{\lambda_0}\otimes Q_{\lambda_0}  )$, we can decompose $\rho_x^{\rm out}$ into orthogonal blocks where  exactly $l$ output systems are  in the sector $\lambda_0$. 
Explicitly, we have 
\begin{align}\label{blocks}
\rho_x^{\rm out}   =   \bigoplus_{l=1}^k      \bigoplus_{\set A  \in  \set S_{l}}      \,  q  (\set A|x )  \, \Big (       \rho_{\set A,x}    \otimes \chi_{\overline{\set A}} \Big) \, ,    
  \end{align}
where  $\set S_{l}$ denotes the set of all  $l$-element subsets  of $\{1,2,\dots, k\}$, $\rho_{x,\set A}$ is the  quantum state defined by
\begin{align}\label{Astates}
\rho_{\set A, x}   & =    \left (   \frac{P_{\lambda_0}}{d_{\lambda_0}}     \otimes      \Phi_{\lambda_0}    \right)_x   \otimes  \left[   \bigotimes_{i\in\set A  \, , i\not  =  x}    \,        \left(  \frac{P_{\lambda_0} }{d_{\lambda_0}}   \otimes \frac{ Q_{\lambda_0} }{m_{\lambda_0}} 
   \right)_i   \right]    \, , 
  \end{align}
  $\chi_{\overline {\set A}}$ is the quantum state defined by
   \begin{align}
\chi_{\overline{\set A}}         &  =  \bigotimes_{i\not \in\set A  }    \, \left(         \frac{      I^{\otimes N}   -      P_{\lambda_0}    \otimes   Q_{\lambda_0}  
   }{d^N   -  d_{\lambda_0}  m_{\lambda_0}}    \right)_i      \, ,  
  \end{align}
  and  $q  (\set A|x)$ is the conditional probability distribution defined by
\begin{align}\label{qA}
q  (\set A|x)   =   \left\{  
 \begin{array}{ll}   p_{\lambda_0}^{l-1}  \,  (1-p_{\lambda_0})^{k-l}    \qquad & {\rm for}~ x\in\set A  \\
 &\\ 
    0     \qquad &  {\rm for}~x\not \in \set A  \, ,  
    \end{array}
    \right.    
\end{align}
 $p_{\lambda}  :=   {  d_\lambda  m_\lambda}/{d^N}$ being the Schur-Weyl measure,

From Eq. (\ref{blocks}) one can see that blocks with  different  values of $l$ and/or different subsets $\set A$ are orthogonal for every value of $x$.   Hence, one can extract first the information about the block and then  the information about $x$.    Mathematically, this means performing a non-demolition measurement with outcomes $(l,  \set A)$, which projects the state into the block  labelled by $(l,\set A)$. 
When such a measurement is performed on the state $\rho_x^{\rm out}$, the outcome $(l,\set A)$ can occur  only if $\set A$ contains $x$---in which case the probability of occurrence is $q (\set A|x)$.   
Conditionally on the outcome, the system is left in the state $\rho_{\set A, x} \otimes \chi_{\overline{\set A}}$ and the problem is to identify $x$ within the set $\set A$.     
Hence, the probability of success for fixed $x$ is 
\begin{align}\label{probA}
p_{\rm succ}  (x)   =      \sum_{l=1}^k      \sum_{\set A  \in  \set S_{l}}      \,  q (\set A|x)    ~     p_{\rm succ}^{(\set A)}  (x)  \, ,     
\end{align}  
where $p_{\rm succ}^{(\set A)}  (x)$ is the probability of correctly identifying the state $\rho_{\set A, x} \otimes \chi_{\overline{\set A}}$.  

Note that,       for $x\in \set A$,  the optimal success probability $p_{\rm succ}^{(\set A)} (x)$ does not depend on the specific subset $\set A$, but only on its cardinality $l$: indeed, $p_{\rm succ}^{(\set A)}  (x)$ coincides with  the probability $p_{\rm succ}^{(l)}  (x)$ of correctly identifying the label of the states 
   \begin{align}\label{lstatescompact}
  \sigma_x   =     \Phi_{x}  \otimes \left(   \frac{ I_m}{m} \right)^{\otimes {l-1}}_{\overline x}    \, ,  \qquad x\in  \{1,2,\dots,  l\} \, ,
  \end{align}
  where  we used the shorthand notation $\Phi_x  :  =  ( \Phi_{\lambda_0})_x$, and used  $I_m$ to denote the identity matrix in dimension $m$, with  $m=  m_{\lambda_0}$ (these are the states that arise from Eq. (\ref{Astates}) after discarding the representation spaces).    We denote  by $p_{\rm succ}^{(l)}$ the average success probability
  \begin{align}
  p_{\rm succ}^{(l)}   =  \frac 1 l  \,  \sum_{x=1}^l  \,  p_{\rm succ}^{(l)}  (x) \, .         
  \end{align}
  
 Averaging the success      probability (\ref{probA}) over $x$, we obtain 
 \begin{align}
\nonumber 
p_{\rm succ}  &   =  \frac 1k  \,\sum_{x=1}^k  \, p_{\rm succ}   (x)  \\
 \nonumber    &   =     \frac 1k  \,\sum_{x=1}^k     \sum_{l=1}^k      \sum_{\set A  \in  \set S_{l}}      \,  q ({\set A}| x)    ~     p_{\rm succ}^{(\set A)}  (x)     \\
  \nonumber    &   =     \frac 1k      \sum_{l=1}^k      \sum_{\set A  \in  \set S_{l}}   \,\sum_{x\in \set A}  \, p_{\lambda_0}^{l-1} \,  (1-p_{\lambda_0})^{k-l}       ~     p_{\rm succ}^{(\set A)}   (x)    \\
   \nonumber    &   =     \frac 1k      \sum_{l=1}^k   \sum_{\set A\in\set S_l}         \, p_{\lambda_0}^{l-1} \,  (1-p_{\lambda_0})^{k-l}     \,  l  \,      p_{\rm succ}^{(l)}              \\
 \nonumber 
 & =  \frac 1k \sum_{l=1}^k   \Big|   \set S_l  \Big| ~  p_{\lambda_0}^{l-1} \,    (1-p_{\lambda_0})^{k-l}  \, l\,   p_{\rm succ}^{(l)}  \\
  \label{probl} & = \frac 1k \sum_{l=1}^k \begin{pmatrix} k \\   l \end{pmatrix}    \,  p_{\lambda_0}^{l-1} \,    (1-p_{\lambda_0})^{k-l}  \, l \,   p_{\rm succ}^{(l)}   \, .  
\end{align}

The next step is to compute $p_{\rm succ}^{(l)}$.

\medskip 

{\em Step 2: reduction to type states. }  The state  $\sigma_x$ in Eq. (\ref{lstatescompact})  is the product of a maximally entangled state  and a  $(l-1)$ copies of the maximally mixed state. 
    The latter can be diagonalized  as  
    \begin{align}
   \left(  \frac{I_m}m\right)_{\overline x}^{\otimes (l-1)}  =  \frac{1}{m^{l-1}} \,   
    \sum_{{\st j} }    |\st j\>\<\st j| \, , 
    \end{align}
  where  $|\st j\>$ is the basis vector  $|\st j\>   =  |j_1\>\otimes |j_2\> \otimes \dots \otimes  |j_{l-1}\>$ corresponding to the sequence     
  ${\st j}  =  (j_1,j_2,\dots,  j_{l-1})   \in  \{1,\dots,  m\}^{\times (l-1)}$.

Now, let us introduce the shorthand   
\begin{align}
|\Phi_{x,\st j\>}     :=   |\Phi\>_{x}  \otimes  |  \st j\>_{\overline x}    \, .
\end{align}
Note that for $x\le y$ one has 
  \begin{align}\label{product}
\<  \Phi_{x, \st j}    |\Phi_{y,\st k}  \>        & =  
\left\{  
\begin{array}{lll}
1  \qquad & x=y \, ,  \quad  &\st j =\st k\\
& & \\ 
\frac 1 m  &  x\not  =  y  \, , \quad      &   j_i  =   k_i  \,,  ~\forall i   <x \\
   &&   j_{i}   =  k_{i+1}  \, , \forall x  \le  i  <  y-1\\
&&  j_{y-1}  =  k_x \\
   &&   j_{i}=   k_i   \, , \forall i\ge y  \\
& &\\ 
0   & & {\rm otherwise}     \, .      
\end{array}
\right.   
\end{align}
 
 Let $\st n  =  (n_1, n_2,  \dots,  n_m)$ be a partition of $l-1$ into $m$ nonnegative integers.  Recall that the sequence $\st j  =  (j_1,j_2,\dots,  j_{l-1})$ is said to be {\em of type $\st n$} if it $n_1$ entries of $\st j$ are equal to $1$, $n_2$ entries are equal to $2$, and so on.   
 Eq. (\ref{product}) tells us that  the vectors $ |\Phi_{x, \st j}\>$ and $|\Phi_{y,\st k}\>$  are orthogonal whenever the sequences $\st j$ and $\st k$ are of different type.   Using this fact, we can define the orthogonal subspaces
 \begin{align}\label{typesubspace}
 \spc H_{\st n}     =  \Span  
 \Big\{    |\Phi_{x,\st j}\>       ~|~     x   \in  \{1,\dots, l \}   \, ,  ~  \st j  \in  \set S_{ \st n}    \, \Big\}  \, ,
 \end{align}
 where $\set S_{\st n}$ is the set of all sequences of length $l-1$ and of type $\st n$. 
 Hence, we can decompose the states $\sigma_x$ in Eq. (\ref{lstatescompact}) as  
\begin{align}\label{blocksagain}
\sigma_x     &=     \bigoplus_{\st n} \,    p(\st n)  \,      \sigma_{\st n ,x}    \, , 
  \end{align}
with 
\begin{align}
p(\st n)  &   =   \frac{  C_{\st n}}{m^{l-1}} \,  \qquad {\rm and}  \qquad  \sigma_{\st n,x}    =        \frac 1 {C_{\st n}}  ~  \sum_{\st j \in \set S_{\st n}}     |\Phi_{x,\st j}\>\<\Phi_{x,\st j}| \, ,
\end{align}
where $C_{\st n}  =  (l-1)!/[  n_1!  n_2!  \cdots  n_m!]$ is the number of sequences of type $\st n$.   
 
Eq. (\ref{blocksagain}) tells us that, in order to distinguish  the states $\sigma_x$, one can perform an orthogonal measurement  that projects on the subspaces $\{\spc H_{\st n}\}$ (\ref{typesubspace}).  If the measurement outcome is $\st n$, one is left with the task of distinguishing among the states $\sigma_{\st n,  x}$.    The success probability  of this strategy is  
\begin{align}\label{psuccstn}
p_{\rm succ}^{(l)}   &   =   \sum_{\st n}\,    p(\st n)  ~    p_{\rm succ}^{(\st n)} \, ,       
\end{align}
where $p_{\rm succ}^{(\st n)}$ is the probability of correctly distinguishing the states $\{\sigma_{\st n,x} ~|~ x  \in  \{1,\dots, l\}  \}$.  

\medskip 

{\em Step 3: lower bound on the probability of success.}    The probability of correctly distinguishing the states $\{\sigma_{\st n,x} ~|~ x  \in  \{1,\dots, l\}  \}$ is lower bounded by the probability of correctly distinguishing among all their eigenstates  
\begin{align}
\Big\{    |\Phi_{x,\st j}\>       ~|~     x   \in  \{1,\dots, l \}   \, ,  ~  \st j \in\set S_{\st n}    \, \Big\}  \, .
\end{align} 
Note that the total number of states is $l    \, C_\st n$. 

We now construct a measurement that distinguishes these states with high success probability. The measurement is constructed through a Grahm-Schmidt orthogonalization procedure.     
We define a first batch of $C_{\st n}$  vectors as  
\begin{align}
|\Psi_{1, \st j}\>   :=   | \Phi_{1,\st j}\>   \qquad \st j  \in  \set S_{\st n} \, . 
\end{align}
This definition is well-posed, because the above vectors are orthonormal, due to Eq. (\ref{product}). 

A second batch of  vectors is constructed from the vectors $\{ |\Phi_{2,\st j}  \>  \, ,   \st j\in  \set S_{\st n}\}$ via the Grahm-Schmidt procedure, which yields 
\begin{align}
|\Psi_{2,\st j}  \>  :   =  \frac {  |\Phi_{2,\st j}\>   -  \frac 1 m  \,   | \Phi_{1,   \st j^{12}}  \> }{\sqrt{1-  \frac 1 {m^2}}} \, , 
\end{align}
where  $\st j^{12}$ is the sequence such that $\<  \Phi_{1,  \st j^{12}} |\Phi_{2,  \st j}  \>  =  1/m$.  

A third batch of  vectors is constructed from the vectors $\{ |\Phi_{2,\st j}  \>  \, ,   \st j\in  \set S_{\st n}\}$.    Now, the Grahm-Schmidt procedure yields 
\begin{align}
|\Psi_{3,\st j}\>    :  =   \frac {  |\Phi_{3,\st j}  \>    -       \frac 1m  \,     |\Phi_{2, \st j^{23}}  \>          -  \frac 1 m  \,  |\Phi_{1, \st j^{13}}\> }{ \sqrt{  1  -   \frac {2}{m^2}}}   +    O  \left(   \frac 1 {m^2}  \right) \,   | \Gamma_{3,\st j}  \>    +  O  \left(   \frac 1 {m^3}  \right)\,  |{\rm  Rest}_{3,\st j} \>   \, ,
\end{align} 
where $| \Gamma_{3,\st j}  \>$ is a vector of the form  $|\Phi_{1,\st k}\>$ for some suitable $\st k$ and   $|{\rm Rest}_{3,\st j}\>$ is a suitable unit vector, which is irrelevant for computing the leading order of the success probability. 

In general, the $x$-th batch of vectors is 
\begin{align}
|\Psi_{x,\st j}\>    :  =   \frac {  |\Phi_{x,\st j}  \>    -       \frac 1m  \,     \sum_{  y  =  1}^{x-1}  |\Phi_{y, \st j^{yx}}  \>  }{\sqrt {1  -  \frac {x-1}{m^2}}}  +    O  \left(   \frac 1 {m^2}  \right) \,   | \Gamma_{x,\st j}  \>    +  O  \left(   \frac 1 {m^3}  \right)\,  |{\rm  Rest}_{x,\st j} \>   \, ,
\end{align} 
where $| \Gamma_{x,\st j}  \>$ is a normalized combination of  vectors of the form  $|\Phi_{z,\st k_z}\>$,  $z<x-2$, while $|{\rm Rest}_{x,\st j}\>$ is a suitable unit vector.   
 
Note that one has  
\begin{align}\label{sudata}
\< \Phi_{x,\st j} |  \Psi_{x,\st j} \>  &  =      \sqrt{  1  - \frac{ x-1}{m^2}}  +    O  \left(   \frac 1 {m^3}  \right)   \, ,  \qquad \forall  x\in  \{1,\dots, l\} \, , \quad \forall \st j\in\set S_{\st n} \, ,
\end{align}
having used the fact that the product $\<\Phi_{x,\set j}|   \Gamma_{x,\st j}\>$ is  $O(1/m)$.  

Using Eq. (\ref{sudata}), we can  now  evaluate the probability of correctly distinguishing the states $  \{  |\Phi_{x,\st j}  \>  \}$.  
 On average over all possible states, the probability of success is 
 \begin{align}
 \nonumber p_{\rm succ}^{(\st n)}  &  =    \frac 1 {  l      C_{\st n}}  \,   \sum_{x=1}^l  \, \sum_{\st j\in\set S_\st n} \,     \,  \Big |  \< \Psi_{x,\st j} |  \Phi_{x,\st j} \>   \Big |^2  \\
 \nonumber  &  =       \frac 1 {  l      C_{\st n}}  \,   \sum_{x=1}^l  \, \sum_{\st j\in\set S_\st n} \,     \left [  1  - \frac{ x-1}{m^2}  +    O  \left(   \frac 1 {m^3}  \right)  \right]    \\
   \nonumber   &  =   \frac 1l   ~   \sum_{x=1}^l   \,     \left [  1  - \frac{ x-1}{m^2}  +    O  \left(   \frac 1 {m^3}  \right)  \right]  \\
   &  =    1   -       \,      \frac{l-1}{2m^2}  +        O  \left(   \frac 1 {m^3}  \right)    \, .    
 \end{align}
 Since measuring on the basis $\{  |\Psi_{x,\st j}\>\}$ is not necessarily the optimal strategy,  we arrived at  the lower bound  
 \begin{align}\label{gsbound}
 p_{\rm succ}^{(\st n)}  \ge      1   -       \,      \frac{l-1}{2m^2}  +        O  \left(   \frac 1 {m^3}  \right)   \, . 
 \end{align}
Note that the (leading order of the) r.h.s. is independent of the type $\st n$.

 \medskip 
 
{\em Step 4: putting everything together.}   Combining the results obtained so far, we can lower bound the success probability in distinguishing among $k$ causal structures.   Inserting the lower bound (\ref{gsbound}) into Eq. (\ref{psuccstn}), we obtain 
\begin{align}
\nonumber p^{(l)}_{\rm succ}  &  =  \sum_{\st n}  \,  p(\st n) ~    p_{\rm succ}^{\st n}  \\
 \nonumber  & \ge   1   -       \,      \frac{l-1}{2m^2}  +        O  \left(   \frac 1 {m^3}  \right)   \, . 
\end{align} 
Then, we can insert the above bound into Eq. (\ref{probl}). Reverting to the full notation $m_{\lambda_0} \equiv m$, we obtain
\begin{align}
\nonumber p_{\rm succ}  &  =  \frac 1 k  \,     \sum_{l=1}^k  \,   \begin{pmatrix}  k  \\   l \end{pmatrix}  \,    p_{\lambda_0}^{l-1} \,  (1-p_{\lambda_0})^{k-l}  \,  l\,    p_{\rm succ}^{(l)}  \\
\nonumber  & \ge     \frac 1 k  \,     \sum_{l=1}^k  \,   \begin{pmatrix}  k  \\   l \end{pmatrix}  \,    l\,    p_{\lambda_0}^{l-1} \,  (1-p_{\lambda_0})^{k-l}   \left[   1   -       \,      \frac{l-1}{2m_{\lambda_0}^2}  +        O  \left(   \frac 1 {m_{\lambda_0}^3}  \right)  \right]   \\
\nonumber  &   =   1  -  \frac {(k-1)  p_{\lambda_0}}{   2  m_{\lambda_0}^2} +   O  \left(   \frac 1 {m_{\lambda_0}^3}  \right)  \\
&   =   1  -  \frac {k-1  }{   2  d^N  }   \frac{d_{\lambda_0}}{m_{\lambda_0}}   +O  \left(   \frac 1 {m_{\lambda_0}^3}  \right)   \, .   
\end{align}
  Hence, the error probability of the optimal quantum strategy is upper bounded as 
  \begin{align}
  p_{\rm err}   \le \frac {k-1  }{   2  d^N  }   \frac{d_{\lambda_0}}{m_{\lambda_0}}  +   O  \left(   \frac 1 {m_{\lambda_0}^3}  \right)\, .
  \end{align}  
Recalling that the ratio $d_\lambda/m_\lambda$ is minimised by the representation with ``minimal" Young diagram (in the majorisation order), we conclude that, when $N$ is a multiple of $d$, the optimal error probability satisfies the bound  
\begin{align}
p_{\rm err}  \le  \frac{  k-1 }{2 d^N  \,  m(N,d)}  +  O  \left(   \frac 1 {m_{\lambda_0}^3}  \right)  \, ,  \qquad {\rm with}  \qquad m(N,d)    =    d^N  \,  \left[ \frac{ d^{\frac{d^2}2}    \, e^{\frac{d(d-1)}2}\,  \prod_{i=1}^d (d-i)!} {(2  \pi)^{\frac{d-1}2} \,  N^{\frac{d^2-1}2}} \right]  \,  c(N) \quad {\rm and} \qquad c(N)  \to 1  \, .
\end{align}
Hence, the asymptotic decay rate is lower bounded as 
\begin{align}
\nonumber R_{\rm Q}  &  =     -  \lim_{N\to \infty}   \frac{\log p_{\rm err}} N  \\
  &  \ge     2 \log d  \, .   
\end{align}
On the other hand,   the r.h.s. is equal to the decay rate for $k=2$, which is a lower bound for the decay rate for $k\ge 2$.  
In conclusion, we obtained that the optimal decay rate is {\em equal} to $R_{\rm Q}  =  2 \log d$.  \qed

\section{Quantum speedup in the identification of a cause}\label{app6}   

We  consider the scenario where  $k$  quantum variables $A_1,\dots, A_k$ are candidate causes of a given effect $B$.  For simplicity, we assume that all  variables  are quantum systems of dimension $d<\infty$. 
The causal relation is described by a quantum channel $\map C_{x, \map U}$ of the form $\map C_{x, \map U}  (\rho ) =    \map U  (  \Tr_{\overline x}[  \rho])$, where  $\Tr_{\overline x}$ denotes the partial trace over all input systems except $A_x$, with $x\in  \{1,\dots, k\}$, and  $\map U$ is a generic unitary channel, acting on the remaining system $A_x$. 
  The problem is to identify the value of $x$.

  \subsection{Fixed unitary gates}
  
  Suppose first that the unitary gate $\map U$ is fixed. Without loss of generality, we can assume $\map U  = \map I$, so that the channel $\map C_{x, \map I}$ is simply the partial trace over all systems except $x$.       The distinguishability of the channels $\{\map C_{ x , \map I} \}_{x=1}^k$
 has been studied extensively in the optimization  of port-based teleportation  \cite{mozrzymas2018optimal}. A simple strategy is to entangle each input system with a reference system, obtaining the output state $\rho_x : = \Phi^+_{B  R_x} \otimes  \left(\frac I d\right)^{\otimes {k-1}}_{\overline x}$, where $\Phi^+$ is the maximally entangled state, $R_x$ is the $x$-th reference system,  and the subscript $\overline x$ indicates that the operator $(I/d)^{\otimes (k-1)}$ acts on the Hilbert space of all reference systems except  $R_x$.  
 
 For $k  \ge d$, the optimal probability of success in distinguishing between the states $\{\rho_x\}_{x=1}^k$ is $p_{\rm succ}  =  d^2/(k-1  +  d^2 )$ \cite{mozrzymas2018optimal}.    If the unknown process is probed for $N$ times, the output state is $\rho_x^{\otimes N}$ and the probability of success is $p_{\rm succ}  =  d^{2N}/(k-1  +  d^{2N} )$.

\subsection{Unknown unitary gates} 
Let us consider the scenario  where  the unitary gate $\map U$ is completely unknown.   By the same argument as in   Appendix 1, the minimum  worst-case error probability is    equal to the minimum error probability in distinguishing between the average channels  
 \begin{align}
 \map C_x^{(N)}  =      \int \d \map U_1 \d \map U_2 \cdots \d \map U_k  \quad  \map C_{x , \map I}^{\otimes N}  \circ \big(    \map U_1 \otimes \map U_2 \otimes \cdots \otimes \map U_k \big)^{\otimes N} \, .
 \end{align}
The symmetry of the problem implies that the optimal input states are of the form  
\begin{align}
\rho_{\bf A R}  =     \frac  { P_{\lambda_1}}{d_{\lambda_1}}   \otimes  \frac  { P_{\lambda_2}}{d_{\lambda_2}}   \otimes \cdots \otimes \frac  { P_{\lambda_k}}{d_{\lambda_k}}   \otimes  \Psi_{  M_{\lambda_1}  M_{\lambda_2} \cdots M_{\lambda_k}  R } \, ,
\end{align}
where $P_{\lambda_i}$  is the projector on the representation space $\spc R_{\lambda_i}$ in the tensor product  $(\spc H^{\otimes N})_i$ of the $N$ systems corresponding to variable $A_i$, and the subscript $M_{\lambda_i}$ denotes the multiplicity space in $(\spc H^{\otimes N})_i$. 

When the input variables are initialized in the state $\rho_{\bf A R}$, the output is 
\begin{align}
\rho_{{\bf B R}, x}  =     \frac  { P_{\lambda_x}}{d_{\lambda_x}}       \otimes \Tr_{\overline {M_{\lambda_x}}}  \left [\Psi_{  M_{\lambda_1}  M_{\lambda_2} \cdots M_{\lambda_k}  R }\right]   \, ,
\end{align}
where $\Tr_{\overline {M_{\lambda_x}}}$ is the trace over all multiplicity spaces except $M_{\lambda_x}$.   

We now show that the true cause can be perfectly identified using at $  O(\log_d  k )$  queries to the unknown process.  We first provide an {\em exact} strategy using $\log_d k$ queries (at the leading order), and then show that the number of queries can be reduced to  $1/2  \log_d k$ (at the leading order) if a small  error, vanishing in the large $k$ limit, is tolerated. 

Our exact strategy  disregards the reference system $R$. In this strategy, we prepare the multiplicity systems in the product state 
\begin{align}|\Psi\>_{  M_{\lambda_1}  M_{\lambda_2} \cdots M_{\lambda_k} }  =    |\psi_1\>_{M_{\lambda_1}} \otimes |\psi_2\>_{M_{\lambda_2}}  \otimes \cdots \otimes |\psi_k\>_{M_{\lambda_k}}  \, .
\end{align}  
 We divide the indices $i$ into   $L$  groups, labelled as $G_1, G_2,  \dots, G_L$ and assign a distinct Young diagram to each group, so that $\lambda_i =  \lambda_j  $ for $i,j$ in the same group.   Within each group, we choose the states $|\psi_i\>_{M_{\lambda_i}}$ to be orthogonal. This choice constrains the number of indices in group $G_l$ to be at most the dimension of the multiplicity space $\map M_{\lambda_{G_l}}$, where $\lambda_{G_l}$ is the Young diagram assigned to the group $G_l$.    In turn, this implies that the condition 
\begin{align}\label{kmin}
k\le \sum_{l=1}^L m_{  \lambda_{G_l}}   \le \sum_\lambda    m_\lambda           
\end{align}  
must be satisfied.   Both bounds can be saturated, as one can choose $L$ to be the number of Young diagrams in the decomposition of the tensor representation $U^{\otimes N}$.  
On the other hand, the multiplicities  are lower bounded as   $m_\lambda   \ge  \begin{pmatrix}   N\\ \lambda \end{pmatrix}/ (N+1)^{d(d-1)/2}$  where $\begin{pmatrix}   N\\ \lambda \end{pmatrix}  =  N! /(\lambda_1!  \lambda_2! \cdots  \lambda_k!)$ is the multinomial coefficient \cite{harrow2005applications,christandl2006spectra}. Hence, we have the bound   $\sum_\lambda  m_\lambda    \ge d^N  / (N+1)^{d(d-1)/2}$, meaning that  condition (\ref{kmin}) can be satisfied with $N \ge  \log_d k  +  O(  \log  \log  k)$.   Hence, the unknown cause can be identified with zero error using approximately $\log_d k$ queries. 

We now construct a strategy that identifies the correct cause with  $1/2  \log_d k  +  O(  \log  \log  k)$ queries and with vanishing error probability.  In this strategy, all the input variables are initialized in the same sector, namely $\lambda_1 =  \lambda_2 = \dots = \lambda_k  \equiv \lambda$.  Specifically, we take $N$ to be a multiple of $d$ and choose $\lambda$ to be the Young diagram corresponding to the trivial representation of $\grp{SU} (d)$. The strategy uses the reference system $R  =    M_\lambda^{\otimes k}$ and the input state 
\begin{align}
\rho_{{\bf A R}}  =    \left(  \frac  { P_{\lambda}}{d_{\lambda}}  \right)^{\otimes k} \otimes \left(\Phi^{+ }_{\lambda }\right)^{\otimes k}\, ,
\end{align}
where $\Phi^+_{\lambda }$ is the projector on the maximally entangled state of two identical copies of $M_\lambda$. Then, the output state is 
\begin{align}
\rho_{{\bf B R}, x}  =    \frac  { P_{\lambda}}{d_{\lambda}}    \otimes \left(\Phi^{+ }_{\lambda }\right)_x \otimes \left (\frac {Q_\lambda}{m_\lambda}  \right)^{\otimes (k-1)}_{\overline x}\, ,
\end{align}
where the maximally entangled state  $\left(\Phi^{+ }_{\lambda }\right)_x$  involves the output system $B$ and the $x$-th reference system, while all the remaining reference systems are in the maximally mixed state $Q_\lambda/m_\lambda$.  Distinguishing among the states $\rho_{{\bf BR} , x}$ is equivalent to distinguishing the states $\left(\Phi^{+ }_{\lambda }\right)_x \otimes \left (\frac {Q_\lambda}{m_\lambda}  \right)^{\otimes (k-1)}_{\overline x}$. This problem has been solved in the context of port-based teleportation, and the minimum error probability is known to be $ p_{\rm err}   =  (k-1)  /  (m_\lambda^2  +  k-1)   $  \cite{mozrzymas2018optimal}.      Using Equation (\ref{asymptoticmulti}), we then obtain  
\begin{align}
p_{\rm err}  \le \frac{k-1}{m_\lambda^2}   =   \frac {k-1}{d^{2N}} \,\left[ \frac{ (2\pi)^{d-1} \,  N^{d^2-1}  }{   d^{d^2}   \,  e^{d(d-1)} \,  \prod_{i=1}^d (d-i)!   \, c(N)  }\right] \, ,
\end{align} 
with $\lim_{N\to \infty}  c(N)  =1$. 
Hence, a vanishing error probability can be obtained by setting $N  = \lceil   (\log_d  k  )(1+\epsilon)/2  \rceil$ with $\epsilon >  0$.

\section{A quantum advantage in the presence of noise}\label{app7}

Here we consider the task of identifying  causal intermediaries when the cause-effect relation is obfuscated by depolarizing noise, corresponding to the channel 
 $\map D_{p}  =   (1-p)  \, \map I  +  p  I/d$, where $p$ is the probability of depolarization.  
 
For simplicity, consider the case of one input variable $A$ and two output variables $B$ and $C$. Suppose that the experimenter  prepares $N$ copies of the maximally entangled state and sends half of each entangled state through one instance of the unknown process.  
 With this choice, the output state consists of $N$ copies of the state $\Sigma_x$, $x  \in \{1,2\}$, with
\begin{align}
\Sigma_1   = \left[ (1-p)  \Phi +  p  \frac{I\otimes I}{d^2}  \right]_{ BR}\otimes \left(\frac{I}d \right)_{ C}  \qquad {\rm and }\qquad \Sigma_2   =  \left(\frac{I}d \right)_{ B} \otimes \left[ (1-p)  \Phi  +  p  \frac{I\otimes I}{d^2}  \right]_{  CR} \, ,
\end{align}
where $\Phi$ is the projector on the canonical maximally entangled state.  Then, the quantum Chernoff bound \cite{audenaert2007discriminating}  yields the rate  
\begin{align}
\nonumber R    &=  - \log  \min_{0\le s \le 1}   \Tr  [  \Sigma_1^s  \Sigma_2^{1-s}]  \\
\nonumber &  = -  \log     \min_{0\le s \le 1}   \frac 1 {d^2}      
 \left[   \left(1-  p    +   \frac p {d^2} \right)^s      +  (d^2-1)  \left(    \frac p{d^2}\right)^s  \, \right] \,  \left[   \left(1-  p    +   \frac p {d^2} \right)^{1-s}      +    (d^2-1) \left(    \frac p{d^2}\right)^{1-s}  \, \right]  \\
 &  =  2\log  d   - 2  \log   \left[ \sqrt{1-  p    +   \frac p {d^2}}  +   \sqrt{ \frac p{d^2}} \right] \, .
\end{align}

When $p$ is small enough,  the rate can be larger than $\log d$, the best classical rate in the noiseless scenario. Since noise can only increase the error probability, this implies a quantum-over-classical advantage in the noisy scenario.  The same result holds for the discrimination of $k\ge 2$ hypotheses, as the quantum Chernoff bound for multiple states is equal to the worst-case Chernoff bound among all pairs  \cite{li2014second}.

We now provide a partial discussion of the scenario where the functional dependence between cause and effect is unknown. This scenario can be modelled by concatenating the depolarizing channel with a completely unknown unitary gate acting on the input variable. The full  analysis of the probability of error is substantially more complex, and we leave it as a topic of  future research.  Here we evaluate the error probability in the simplified scenario where the depolarization is heralded, meaning that when the system is depolarized to the maximally mixed state, the process outputs a classical outcome.    Taking this piece of information into account,   the error probability becomes  $p_{\rm err}  =  \sum_{k=0}^N  \,  (1-p)^k \, p^{N-k}  \begin{pmatrix}  N\\  k \end{pmatrix}  \,  p_{{\rm err} , k}$, where $p_{{\rm err}, k}$ is the probability of error with $k$ noiseless experiments.  
 
 The evaluation of $p_{{\rm err}, k}$ is as follows.  The input state of $k$ maximally entangled states, averaged over all possible unitary gates is 
 \begin{align}
 \rho_{\rm in}   =  \bigoplus_\lambda \,  p_{\lambda}  \, \frac{  P_\lambda\otimes P_\lambda}{d_{\lambda}^2}  \otimes  \Phi_\lambda  \, ,
 \end{align}   
  where $p_{\lambda}  =  d_\lambda m_\lambda/d^N$ is the Schur-Weyl measure, and $\Phi_\lambda$ is the maximally entangled state in $\spc M_\lambda\otimes \spc M_\lambda$. 
  
 The two output states corresponding to the two hypotheses are 
 \begin{align}
 \rho_{\rm out, 1}  =  (\rho_{\rm in})_{\bf BR} \otimes \left ( \frac I d \right)^{\otimes k}_{\bf C} \qquad {\rm and} \qquad  \rho_{\rm out, 2}  =\left ( \frac I d \right)^{\otimes k}_{\bf B}  \otimes   (\rho_{\rm in})_{\bf CR}  \, . 
 \end{align}
 
  The distance between them is 
  \begin{align}
 \nonumber \|  \rho_{\rm out, 1}   -   \rho_{\rm out, 2} \|_1   &  =  \left \|   \bigoplus_{\lambda, \mu, \nu}       \frac{  (P_\lambda)_{\bf B}\otimes  (P_\mu)_{\bf R}\otimes (P_\nu)_{\bf C}}{d_\lambda d_\mu d_\nu}  \otimes \left[  p_\lambda   p_\nu   \delta_{\lambda\mu}   (\Phi_\lambda)_{\bf BR} \otimes \left( \frac {Q_\nu}{m_\nu}\right)_{\bf C}    -   p_\lambda p_\mu  \delta_{\mu \nu}  \left(   \frac{Q_\lambda}{m_\lambda}  \right)_{\bf B}   \otimes  (\Phi_\mu)_{\bf RC}    \right]      \right\|_1 \\
\nonumber  &  =  2 \left( 1  -  \sum_\lambda   \, p_\lambda^2\right)  +  \sum_\lambda   p_\lambda^2  \left\|      (\Phi_\lambda)_{\bf BR} \otimes \left( \frac {Q_\lambda}{m_\lambda}\right)_{\bf C}    -  \left(   \frac{Q_\lambda}{m_\lambda}  \right)_{\bf B} \otimes (\Phi_\lambda)_{\bf CR}    \right\|_1  \\
\nonumber  &  =  2 \left( 1  -  \sum_\lambda   \, p_\lambda^2\right)  +2   \sum_\lambda   p_\lambda^2         \sqrt{ 1-m_\lambda^{-2}}  \, ,
  \end{align} 
  where the second term in the sum has been evaluated through Equation (\ref{crucial}).

Hence, we have the approximate expression  
  \begin{align}
 \nonumber \|  \rho_{\rm out, 1}   -   \rho_{\rm out, 2} \|_1  
 &  =  2 \left( 1  -  \sum_\lambda   \, p_\lambda^2\right)  +2   \sum_\lambda   p_\lambda^2       \left[1   -  \frac 1 {2 m_\lambda^2}  +  O(m_\lambda^4)\right]\\
 \nonumber &  =    2  -  \frac 1{d^{2 k}} \,\sum_\lambda   d^2_\lambda    +  O(  d^{-4k}) \\
&  =    2  -  \frac {{\rm Poly}  (k,d)}{d^{2k}}  +  O(  d^{-4k}) \, ,
  \end{align}
  having used the fact that the dimensions and the number of Young diagrams grow at most polynomially in $k$ and $d$  (see e.g. \cite{harrow2005applications,christandl2006spectra}).  
    Using the above expression, we obtain the expression $p_{{\rm err},  k}  = \frac {{\rm Poly}  (k,d)}{ d^{2k}}  + O(  d^{-4k})  $.  Summing over $k$ and averaging with the Bernoulli distribution we finally obtain $p_{\rm err}   =  {\rm Poly}  (N,d)  \,  \left(   \frac {1-p}{d^2}  +  p \right)^N     $ at the leading order.   
    
    In conclusion, the discrimination rate is $R  =  -  \log  \left(   \frac {1-p}{d^2}  +  p \right)$, which is larger than the  noiseless classical rate $\log d$  when $p$ is smaller than $1/(d+1)$.    The rate $R  =  -  \log  \left(   \frac {1-p}{d^2}  +  p \right)$ provides an upper bound to the achievable rate without heralding, for the simple strategy consisting in preparing $N$ copies of the maximally entangled state. When the probability of depolarisation exceeds $1/(d+1)$ this simple quantum strategy cannot beat the noiseless classical rate, and  comparison between quantum and classical strategies requires a more detailed analysis. 
     
It is worth noting  the above derivation provides an alternative strategy for the identification of the causal intermediary  in the noiseless case ($p=0$). 
    This strategy achieves the same rate of our universal strategy, although with a polynomially worse error probability.   While suboptimal, the present strategy is practically interesting because it does not require input states with  large-scale multipartite entanglement. 

\section{Proof of Equation (29) in the main text}\label{app8}

{\em Step 1.} Observe that the  channels $\map C_\pm   =    \frac 2 {d^N \pm 1}  \,  P_\pm  \left(\rho \otimes I^{\otimes N}  \right)  \,  P_\pm $ are no-signalling. Indeed, for every subset  $\set S  \subseteq \{1,\dots,  N\}$ one has that the input system ${\bf A} (S)   :=   \bigotimes_{  i \in \set S}  A_i $ cannot signal to the output system ${\bf BC}  (\overline {\set S})  : =  \bigotimes_{  i \not \in \set S}    (B_i \otimes C_i)$.   To check the no-signalling condition, we use the relation 
\begin{align}
P_\pm  =    \frac{I\pm   \text{SWAP} }2 =     \frac{   \left(     \bigotimes_{i\in\set S}   I_{  B_i C_i} \right)  \otimes   \left(   \bigotimes_{i\not \in \set S}   I_{B_i C_i}  \right)  \,  \pm  \,  \left(   \bigotimes_{i\in\set S}   \text{SWAP}_{  B_i C_i} \right)  \otimes \left(   \bigotimes_{i\not \in \set S}   \text{SWAP}_{B_i C_i}   \right)}2 \,,
\end{align} where  $I_{B_iC_i}$ is the identity operator on the composite system $B_i C_i$, and  $\text{SWAP}_{B_i C_i}$ is the unitary operator that swaps systems $B_i$ and $C_i$.  The  state of the output system ${\bf BC}  (\overline {\set S})$ is 
\begin{align}
\nonumber \left(  \bigotimes_{i\in \set S}  \Tr_{B_i C_i} \right)  [  \map C_{\pm}  (\rho)]  &  \propto  d^{|\set S|}\left(  \bigotimes_{i\in \set S}  \Tr_{B_i} \right)  \Big[\, \rho \, \Big]   \otimes \left(  \bigotimes_{i\not \in \set S}  I_{C_i} \right)   + d^{|\set S|} \left(  \bigotimes_{i\not \in \set S}  I_{C_i} \right) \otimes  \left(  \bigotimes_{i\in \set S}  \Tr_{B_i} \right)  \Big[\,  \rho\,\Big]  \\
\nonumber 
&  \quad  \pm \left[  \left(  \bigotimes_{i\in \set S}  \Tr_{B_i} \right)  \Big[\, \rho \, \Big]   \otimes    \left(  \bigotimes_{i\not \in \set S}  I_{C_i} \right)  \right]  \,   \left(   \bigotimes_{i\not \in \set S}   \text{SWAP}_{B_i C_i}   \right)\\
&  \quad  \pm    \left(   \bigotimes_{i\not \in \set S}   \text{SWAP}_{B_i C_i}   \right)  \,
 \left[  \left(  \bigotimes_{i\in \set S}  \Tr_{B_i} \right)  \Big[\, \rho \, \Big]   \otimes    \left(  \bigotimes_{i\not \in \set S}  I_{C_i} \right)  \right] \end{align}   
and depends only on the state of the input system ${\bf A}  (\overline {\set S})$. \\

{\em Step 2.}  Show that there exist coefficients $a$ and $b$ such that the maps $  a \,  \map C_+ +  b\, \map C_-  -  1/2  \,  \map C_{1, I}$ and $  a \, \map C_+ +  b\, \map C_-  -  1/2  \,  \map C_{2,I}$ are completely positive. 

Let us consider the $N=1$ case first. By definition, one has 
\begin{align}
 a \,  \map C_+ +  b\, \map C_-  -  1/2  \,  \map C_{1, I}     =      \map M  \circ (  \map I\otimes I) \qquad {\rm and}  \qquad
a \,  \map C_+ +  b\, \map C_-  -  1/2  \,  \map C_{1, I}      =      \map M  \circ (  I\otimes \map I) \, ,
\end{align}
where $\map M$ is the linear map defined by
\begin{align}  
\map M (A)   :=   &  m_{00}    \,    A    +   m_{01}  \,  A  \, (\text{SWAP})   +   m_{10} \,  (\text{SWAP})  \, A          +  m_{11} \,   (\text{SWAP}) \, A  \, (\text{SWAP}) 
    \end{align} 
    with 
    \begin{align}  
\nonumber    m_{00} &=    \frac  a{2(d+1)}  +  \frac  b{2(d-1)}    -  \frac 1{2d}       \qquad     & m_{10} &=   \frac  a{2(d+1)}  -  \frac  b{2(d-1)}\\
  m_{10}  &=  \frac  a{2(d+1)}  -  \frac  b{2(d-1)}     & m_{11}&=  
  \frac  a{2(d+1)}  +  \frac  b{2(d-1)}  \, .
\end{align}  
Now,  if the matrix $ M$ is positive, then the map $\map M$ is completely positive.  Defining 
\begin{align}
\alpha  :  =  \frac  a{2(d+1)} \qquad {\rm and}   \qquad  \beta  : = \frac  b{2(d-1)} \, ,
 \end{align}
the positivity condition becomes 
\begin{align}
\label{lin}\alpha+  \beta &\ge 1/(2d) \\
\label{quad}4\alpha \beta   & \ge  (\alpha  + \beta)/(2d) \, .    
\end{align}
As an ansatz, we choose $\alpha  =  \sqrt{d-1}  \, x$ and $\beta  =  \sqrt{d+1} \, x$, for some $x  > 0$.  
Then,  condition (\ref{quad}) becomes 
\begin{align}
x &  \ge  \frac 1{8 d} \,  \left( \frac 1 { \sqrt {d+1}}  + \frac 1 { \sqrt{d-1}} \right)  = :  x_0 \, .
\end{align}
Note that the choice $x=  x_0$ satisfies both conditions (\ref{quad}) and  (\ref{lin}). Finally, note that the above derivation holds for arbitrary $N$, by replacing $d$ with $d^N$. \\
 
 {\em Step 3.}  Define the constant $\lambda:  =  a+b$ and the no-signalling channel $\map C :  = ( a \,  \map C_+   +   b \,  \map C_-  )/\lambda$.  By construction,  the maps $\lambda \,\map C  - 1/2  \,  \map C_{1,I}$ and $\lambda \,\map C  - 1/2  \,  \map C_{2,I}$ are completely positive. 
Explicit evaluation yields 
\begin{align}
\lambda  
  &  = \frac{ \left( \sqrt{d^N+1} +  \sqrt {d^N-1}  \right)^2 }{4d^N} \, .   
\end{align} 
Finally, observe that the maps $\lambda \,\map C  - 1/2  \,  \map C_{1,I}$ and $\lambda \,\map C  - 1/2  \,  \map C_{2,I}$ are completely positive if and only if the Choi operators $C, C_{1,I}$, and $C_{2,I}$ corresponding to $\map C, \map C_{1,I}$, and $\map C_{2,I}$ satisfy the inequalities $\lambda \,  C  \ge 1/2 \,  C_{1,I}$ and $\lambda\, C \ge  1/2 \, C_{2,I}$.   Inserting the expression of $\lambda$ into Equation (26) of the main text, we then obtain the desired bound 
\begin{align}
p^{\rm ind}_{\rm err}  \ge  1-\lambda  =  \frac  {  1 -   \sqrt{1-\frac 1 {d^{2N}}} } {2} \, .
\end{align}
 \end{widetext}

\end{document}